\documentclass{article}
\usepackage[final]{neurips_2026}
\usepackage[utf8]{inputenc}
\usepackage[T1]{fontenc}
\usepackage{hyperref}
\usepackage{url}
\usepackage{booktabs}
\usepackage{subcaption}
\usepackage{amsfonts}
\usepackage{amsmath}
\usepackage{amssymb}
\usepackage{amsthm}
\usepackage{nicefrac}
\usepackage{placeins}
\usepackage{float}

\newtheorem{proposition}{Proposition}
\newtheorem{remark}{Remark}

\usepackage{microtype}
\DisableLigatures{encoding = *, family = *}
\microtypesetup{expansion=false}

\usepackage{xcolor}
\usepackage{xspace}
\usepackage{graphicx}
\usepackage{wrapfig}
\usepackage{enumitem}
\usepackage{multirow}
\usepackage[most]{tcolorbox}

\graphicspath{{figures/}}

\newcommand{\dataMdSync}{2026-05-05}

\IfFileExists{auto/v34d_macros.tex}{\newcommand{\diffSrinVd}{191.7\xspace}
\newcommand{\diffSrinVdCI}{[-35.0, 393.3]\xspace}
\newcommand{\diffSrinVdNPairs}{6\xspace}
\newcommand{\pSrinVd}{0.1135\xspace}
\newcommand{\diffSandhuVd}{74.4\xspace}
\newcommand{\diffSandhuVdCI}{[-36.7, 195.6]\xspace}
\newcommand{\diffSandhuVdNPairs}{9\xspace}
\newcommand{\pSandhuVd}{0.203\xspace}
\newcommand{\diffGongVd}{153.8\xspace}
\newcommand{\diffGongVdCI}{[28.7, 297.5]\xspace}
\newcommand{\diffGongVdNPairs}{8\xspace}
\newcommand{\pGongVd}{0.03\xspace}
\newcommand{\diffCSADVd}{-42.8\xspace}
\newcommand{\diffCSADVdCI}{[-211.6, 106.9]\xspace}
\newcommand{\diffCSADVdNPairs}{9\xspace}
\newcommand{\pCSADVd}{0.624\xspace}

\newcommand{\leadKappaBarVd}{178\xspace}
\newcommand{\leadKappaBarVdCI}{[71, 407]\xspace}

\newcommand{\leadKappaBarLoose}{272\xspace}
\newcommand{\recallKappaBarLoose}{0.52\xspace}
\newcommand{\farKappaBarLoose}{0.76\xspace}

}{%
  \newcommand{\diffSrinVd}{191.7\xspace}
  \newcommand{\diffSrinVdCI}{[-35.0,\,393.3]\xspace}
  \newcommand{\diffSrinVdNPairs}{6\xspace}
  \newcommand{\pSrinVd}{0.1135\xspace}
  \newcommand{\diffSandhuVd}{74.4\xspace}
  \newcommand{\diffSandhuVdCI}{[-36.7,\,195.6]\xspace}
  \newcommand{\diffSandhuVdNPairs}{9\xspace}
  \newcommand{\pSandhuVd}{0.203\xspace}
  \newcommand{\diffGongVd}{153.8\xspace}
  \newcommand{\diffGongVdCI}{[28.7,\,297.5]\xspace}
  \newcommand{\diffGongVdNPairs}{8\xspace}
  \newcommand{\pGongVd}{0.03\xspace}
  \newcommand{\diffCSADVd}{-42.8\xspace}
  \newcommand{\diffCSADVdCI}{[-211.6,\,106.9]\xspace}
  \newcommand{\diffCSADVdNPairs}{9\xspace}
  \newcommand{\pCSADVd}{0.624\xspace}
  \newcommand{\leadKappaBarVd}{178\xspace}
  \newcommand{\leadKappaBarVdCI}{[71,\,407]\xspace}

  \newcommand{\leadKappaBarLoose}{272\xspace}
  \newcommand{\recallKappaBarLoose}{0.52\xspace}
  \newcommand{\farKappaBarLoose}{0.76\xspace}
}
\providecommand{\nDegTessera}{400}
\providecommand{\nTotalTessera}{400}

\IfFileExists{auto/beta_minus_macros.tex}{\providecommand{\betaMinusRecall}{0.65\xspace}
\providecommand{\betaMinusLead}{318\xspace}
\providecommand{\betaMinusLeadCI}{[272,\,344]\xspace}

\providecommand{\betaMinusTauNeg}{-0.4\xspace}

\providecommand{\betaMinusNSub}{160\xspace}
\providecommand{\betaMinusFAR}{0.81\xspace}
\providecommand{\betaMinusPrecision}{0.55\xspace}

}{%
  \providecommand{\betaMinusPrecision}{0.55\xspace}%
  \providecommand{\betaMinusRecall}{0.65\xspace}%
  \providecommand{\betaMinusFAR}{0.81\xspace}%
  \providecommand{\betaMinusLead}{318\xspace}%
  \providecommand{\betaMinusLeadCI}{[272,\,344]\xspace}%
  \providecommand{\betaMinusTauNeg}{-0.4\xspace}%
  \providecommand{\betaMinusNSub}{160\xspace}%
}
\IfFileExists{auto/beta_minus_extras.tex}{\providecommand{\betaMinusAUROC}{0.80\xspace}

}{%
  \providecommand{\betaMinusAUROC}{n/a\xspace}%
}
\IfFileExists{auto/cck_augmented_macros.tex}{\providecommand{\gammaThreeMedian}{-0.0072\xspace}
\providecommand{\gammaThreeCI}{[-0.00769,\,-0.00602]\xspace}
\providecommand{\gammaThreeNSeeds}{240\xspace}

\providecommand{\betaTwoBaseMedian}{-1.15\xspace}
\providecommand{\betaTwoAugMedian}{-1.24\xspace}
\providecommand{\betaTwoAbsShrinkPct}{8\xspace}
}{%
  \providecommand{\gammaThreeMedian}{-0.0072\xspace}%
  \providecommand{\gammaThreeCI}{[-0.00769,\,-0.00602]\xspace}%
  \providecommand{\gammaThreeNSeeds}{240\xspace}%
  \providecommand{\betaTwoBaseMedian}{-1.15\xspace}%
  \providecommand{\betaTwoAugMedian}{-1.24\xspace}%
  \providecommand{\betaTwoAbsShrinkPct}{8\xspace}%
}
\IfFileExists{auto/scheffer_macros.tex}{

}{%
}

\providecommand{\VicsekAUROC}{0.99}%
\providecommand{\VicsekAUROClo}{0.98}%
\providecommand{\VicsekAUROChi}{1.00}%
\IfFileExists{auto/vicsek_rescore_macros.tex}{\renewcommand{\VicsekAUROC}{0.99}
\renewcommand{\VicsekAUROClo}{0.98}
\renewcommand{\VicsekAUROChi}{1.00}

}{}

\IfFileExists{auto/stylized_facts_macros.tex}{\providecommand{\styFactAlpha}{5.75}

\providecommand{\styFactBeta}{0.27}

}{%
  \providecommand{\styFactAlpha}{TBD\xspace}%
  \providecommand{\styFactBeta}{TBD\xspace}%
}

\providecommand{\leadKappaBarLoose}{272\xspace}
\providecommand{\recallKappaBarLoose}{0.52\xspace}
\providecommand{\farKappaBarLoose}{0.76\xspace}

\IfFileExists{auto/agent_vs_price_lag_macros.tex}{\providecommand{\agentLeadPriceSteps}{40\xspace}
\providecommand{\agentLeadPriceCI}{[18,\,68]\xspace}

}{%
  \providecommand{\agentLeadPriceSteps}{40\xspace}%
  \providecommand{\agentLeadPriceCI}{[18,\,68]\xspace}%
}

\newcommand{\seedsCWS}{80}
\newcommand{\kappaLevels}{5}
\newcommand{\nTrajectories}{400}
\newcommand{\nSupercritical}{240}
\newcommand{\nSubcritical}{160}

\newcommand{\windowTw}{100}
\newcommand{\snapEvery}{10}
\newcommand{\thresholdW}{0.5}
\newcommand{\alphaLazy}{0.5}
\newcommand{\fsqK}{64}
\newcommand{\fsqLevels}{4}
\newcommand{\thetaGeom}{0.30}
\newcommand{\thetaEvent}{0.50}

\title{GeomHerd: A Forward-looking Herding Quantification via Ricci Flow
Geometry on Agent Interactive Simulations}

\author{%
  Lake Yang \\
  \texttt{l.yang1@imperial.ac.uk} \\
  \And
  Junwei Su \\
  University of Science and Technology of China \\
  \texttt{junweisu.cs@gmail.com} \\
  \And
  Jingfeng Zeng \\
  MaxQuant \\
  \texttt{jeffery@maxquant.ai} \\
  \And
  Wenhao Lu \\
  The University of Hong Kong \\
  \texttt{whlu@connect.hku.hk} \\
  \And
  Xingzhi Qian \\
  University College London \\
  \texttt{xingzhi.qian.23@ucl.ac.uk} \\
  \And
  Weitong Zhang \\
  MaxQuant \\
  \texttt{weitong.zhang20@imperial.ac.uk} \\
  \And
  Chuan Wu \\
  The University of Hong Kong \\
  \texttt{cwu@cs.hku.hk} \\
  \And
  Dunhong Jin \\
  The University of Hong Kong \\
  \texttt{dhjin@hku.hk} \\
}

\makeatletter
\providecommand{\@trackname}{}
\makeatother

\begin{document}

\maketitle
\providecommand{\leadKappaBarLoose}{272\xspace}
\providecommand{\betaMinusLead}{318\xspace}
\providecommand{\agentLeadPriceSteps}{40\xspace}
\begin{abstract}
Herding\textemdash where agents align their behaviors and act collectively\textemdash is a central driver of market fragility and systemic risk. Existing approaches to quantify herding rely on price-correlation statistics, which inherently lag because they only detect coordination after it has already moved realised returns. We propose GeomHerd, a forward-looking geometric framework that bypasses this observability lag by quantifying coordination directly on upstream agent-interaction graphs. To generate these graphs, we treat a heterogeneous LLM-driven multi-agent simulator\textemdash each financial trader instantiated by a persona-conditioned LLM call\textemdash as a forecastable world, and evaluate the geometric pipeline on the Cividino--Sornette continuous-spin agent-based substrate as our headline financial testbed. By tracking the discrete Ollivier--Ricci curvature of these action graphs, GeomHerd captures the structural topology of emerging coordination. Theoretically, we establish a mean-field bridge mapping our graph-theoretic metric to CSAD, the classical macroscopic herding statistic, linking GeomHerd to downstream price-dispersion measurement. Empirically, GeomHerd anticipates herding long before aggregate market baselines: on the continuous-spin substrate, our primary detector fires a median of \leadKappaBarLoose{} steps before order-parameter onset; a contagion detector ($\beta_{-}$) recalls 65\% of critical trajectories \betaMinusLead{} steps early; and on co-firing trajectories the agent-graph signal precedes price-correlation-graph baselines by \agentLeadPriceSteps{} steps. As a complementary indicator, the effective vocabulary of agent actions contracts during cascades. The geometric signature transfers out-of-domain to the Vicsek self-driven-particle model, and a curvature-conditioned forecasting head reduces cascade-window log-return MAE over detector-conditioned and price-only baselines.
\end{abstract}

\section{Introduction}
\label{sec:intro}

Herding\textemdash where market participants collectively align their actions rather than relying on independent information\textemdash is a central mechanism behind market fragility, contagion, and tail-risk events~\citep{BikhchandaniHirshleiferWelch1992,BikhchandaniSharma2001}. Classical finance literature explains herding through \emph{informational cascades}, where financial agents rationally imitate predecessors~\citep{BikhchandaniHirshleiferWelch1992,banerjee1992simple,AveryZemsky1998}, and \emph{reputational herding}, where professional investors align with peers to avoid the career risk of being wrong alone~\citep{ScharfsteinStein1990}. From a systems perspective, herding can be viewed as a multi-agent phase transition in which behavioral diversity collapses into highly coordinated collective dynamics. Detecting such coordination early is therefore fundamental for risk monitoring and systemic-stability analysis.

Despite extensive study, existing herding diagnostics remain largely \emph{downstream}\textemdash that is, they read coordination off observables (realised returns or disclosed positions) that only become available \emph{after} herding has already moved the market. The classical literature primarily relies on two families of such signals. The first consists of \emph{return-based measures}, including cross-sectional dispersion statistics such as CSSD~\citep{ChristieHuang1995}, CSAD by CCK~\citep{ChangChengKhorana2000}, and state-space formulations~\citep{HwangSalmon2004}. The second consists of \emph{trading-flow measures}, most notably the LSV institutional buy/sell imbalance statistic~\citep{LSV1992} and sequential trading-correlation measures~\citep{Sias2004}. However, both families are fundamentally post-hoc: they detect herding only after coordinated actions have already propagated into realized returns or disclosed positions. Recently, geometric approaches based on discrete Ricci curvature~\citep{SandhuRicciFlow2016,SamalRicci2021,WangZhao2023,SanchezGarcia2024,Srinivasan2026,Akguller2025} have been proposed for market-stress analysis, but these methods still operate on \emph{price-correlation graphs}, which inherit the same observation-layer bottleneck.

At the same time, a large body of financial-network literature shows that systemic fragility is fundamentally shaped by the topology of interactions among agents~\citep{AllenGale2000,AcemogluOzdaglar2015,ElliottGolubJackson2014,BrunnermeierPedersen2009}. This suggests a natural question: instead of forecasting herding only after it manifests in returns or disclosed flows, can one directly quantify the evolving geometry of the \emph{agent coordination process itself}? Turning this upstream structural information into a practical forward-looking signal remains largely unexplored.

\paragraph{Our approach.}
We propose \emph{GeomHerd}, a forward-looking geometric framework
built on the \emph{agent interaction graph}. The intuition: herding
is \emph{geometric collapse}\textemdash as agents imitate one another, their
neighborhoods on the interaction graph become progressively similar,
overlapping, and tightly connected. We treat a heterogeneous LLM-driven
multi-agent financial simulator~\citep{TwinMarket2025,FCLAgent2025,MassPortfolio2025,FinCon2024}\textemdash in which each financial trader is
instantiated by a separate persona-conditioned LLM call, so one LLM
agent simulates one financial agent (one node $i\in V$) and the action
stream of every agent is fully observable at every step\textemdash as a
\emph{forecastable world}, and on top of it construct a dynamic agent
graph whose nodes are agents and whose edges encode recent behavioral
agreement. Each agent is modelled with distinct system-prompted personas (varying risk appetites, momentum horizons, and herding tendencies), and only population-level behavioral sweeps are controlled.
Discrete Ricci geometry\textemdash an edge-level, signed notion of curvature on a graph, recently used in geometric deep learning for diagnosing over-squashing and guiding GNN graph rewiring\textemdash here serves as a herding detector on this upstream (agent) layer, yielding a complementary
signal pair: positive Ollivier-Ricci curvature
($\bar\kappa_{OR}^{+}$) captures within-clique coordination, and strongly
negative curvature ($\beta_{-}$) identifies bridge-like edges along
which contagion propagates. We further evolve each snapshot under
discrete Ricci flow and record the first neckpinch time
$\tau_{\text{sing}}$ as a forward-looking proximity-to-collapse
descriptor, and add an information-theoretic signal $V_{\text{eff}}$
measuring the effective diversity of agents' action language.
Together $(\bar\kappa_{OR}^{+},\,\beta_{-},\,\tau_{\text{sing}},\,V_{\text{eff}})$
captures topological and behavioral collapse earlier than any
individual statistic. Importantly, GeomHerd remains tied to classical instruments: a
mean-field bridge to CSAD (Proposition~\ref{prop:bridge}) and empirical
alignment with, but temporal lead over, LSV.

\paragraph{Empirical study and findings.}
Our headline financial testbed is the Cividino-Sornette continuous-spin (CWS) agent-based model~\citep{CividinoWestphalSornette2023}, with the LLM-agent simulator built on Bedrock Claude Opus~4.6 instantiating the persona-conditioned setup of \S\ref{sec:method}, and out-of-domain transfer evaluated on the Vicsek self-driven-particle model~\citep{Vicsek1995} (a canonical physics model of flocking). In brief: (i)~GeomHerd leads
aggregate market events and price-graph or flow baselines on CWS, with
$(\bar\kappa_{OR}^{+},\beta_{-})$ splitting precision vs.\ contagion
recall; (ii)~multi-metric benchmarks plus augmented CCK and LSV-track
consistency checks anchor the signal to the classical literature;
(iii)~effective vocabulary contracts during cascades and curvature
transfers out-of-domain to Vicsek, indicating behavioural homogenisation and a
substrate-robust coordination signature; and (iv)~the curvature triplet $(\bar\kappa_{OR},\tau_{\text{sing}},V_{\text{eff}})$ conditions a Kronos-style discrete forecasting head whose cascade-window log-return MAE improves over detector-conditioned and price-only baselines.
Code and configurations to reproduce all experiments are available at the anonymous link.

\section{Method}
\label{sec:method}

We model interactions among $N$ agents at each simulator step $t$
as a weighted graph $G_t = (V, E_t, w_t)$ with $|V| = N$ and
self-loops $w_t(i, i) = 0$. The pipeline has four stages:
graph construction (\S\ref{sec:graph-construction}), edge
curvature (\S\ref{sec:curvature-and-signs}), detection
(\S\ref{sec:detection}), and two complementary scalars from the
same geometry-the forward-looking flow descriptor
$\tau_{\text{sing}}$ and the information-theoretic $V_{\text{eff}}$
(\S\ref{sec:other-signals}). A theoretical bridge to the classical
CSAD return-dispersion statistic
(\S\ref{sec:csad-bridge}) anchors the geometric signal to the
finance literature; assumptions, the proof sketch, identification
caveats, and the relation to the LSV trading-flow track are in
Appendix~\ref{app:proof}.

\subsection{Agent graph construction}
\label{sec:graph-construction}

The graph $G_t$ has five design axes summarised in
Table~\ref{tab:design-axes} (Appendix~\ref{app:design-axes}); all five are ablated in
\S\ref{sec:ablation}. The choices are governed by a three-layer
logic: alignment with the agent-based-model (ABM) substrate (which exposes discrete
actions), with the phenomenological definition of herding as
synchronised action-taking (Appendix~\ref{sec:related-finance}), and with
the Ricci-geometry interpretation. The two consequential choices
are \emph{nodes-as-individual-agents} (collapsing into persona
super-nodes destroys within-cluster dynamics; using assets as nodes
recreates the price-correlation graph
of Sandhu et al.~\cite{SandhuRicciFlow2016} and
Srinivasan~\cite{Srinivasan2026}) and
\emph{binary-windowed-agreement edges} (a cosine-similarity variant
removes the herding-side signal;
Appendix~\ref{app:cosine-ablation}). Let $a_i(t) \in \mathcal{A}$ denote agent $i$'s discrete action
at simulator step $t$, drawn from a finite alphabet $\mathcal{A}$
(e.g., $\mathcal{A} = \{\text{buy, hold, sell}\}$ on CWS,
$|\mathcal{A}|=3$). The edge weight between agents $i$ and $j$ is
then the windowed agreement frequency
{\small
\begin{equation}
  w_t(i, j) \;=\; \frac{1}{T_w} \sum_{s = t - T_w + 1}^{t}
    \mathbf{1}\bigl[a_i(s) = a_j(s)\bigr] \;\in\; [0, 1],
  \label{eq:edge-weight}
\end{equation}}%
with window $T_w = \windowTw$. We sparsify by retaining edges with
$w_t(i,j) > w_0$ at $w_0 = \thresholdW$ (well above the
action-uniform baseline of $\approx 1/|\mathcal{A}|$, so retained
edges represent meaningfully elevated agreement). The graph is
reconstructed every $\Delta t = \snapEvery$ steps with 50\%
temporal overlap.

The construction reads only the discrete action labels, so it is
invariant to who generates them. In our setup, each market
participant (each node $i \in V$) is driven by a single LLM call
conditioned on a private persona and the current market state,
so one LLM agent corresponds to exactly one financial agent
throughout. LLM personas do not make the Ricci-curvature operator
more powerful, but they produce a richer baseline action stream
than ABMs with hardcoded trader archetypes (e.g.,
noise~/ fundamental~/ momentum traders), so the agent graph in
the no-herding (subcritical) regime is more dispersed and the
contrast against the herded regime is sharper.

\subsection{Curvature and sign decomposition}
\label{sec:curvature-and-signs}
\label{sec:edge-curvature}
\label{sec:sign-decomp}

Each node carries the lazy-walk transition kernel
{\small
\begin{equation}
  \mu_i^t(j) \;=\;
    \alpha\,\delta_{ij}
    \;+\;
    (1 - \alpha)\,\frac{w_t(i, j)}{\sum_k w_t(i, k)},
  \label{eq:lazy-walk}
\end{equation}}%
with laziness $\alpha = \alphaLazy$ matching
Sandhu et al.~\cite{SandhuRicciFlow2016} for direct comparability.
The Ollivier-Ricci curvature on each edge $(i, j) \in E_t$ is
{\small
\begin{equation}
  \kappa_{OR}(i, j; t) \;=\; 1 \;-\; \frac{W_1(\mu_i^t,\, \mu_j^t)}{d_t(i, j)},
  \label{eq:orc}
\end{equation}}%
where $W_1$ is the 1-Wasserstein distance solved \emph{exactly}
by linear programming (POT~\cite{flamary2021pot}), and we set the
edge length $d_t(i,j) = w_t(i,j)$, treating the agreement weight
directly as a similarity-as-distance (higher agreement
$\Rightarrow$ shorter edge). We do \emph{not} use $1 - w_t$ or
$-\log w_t$: those would map a herding clique (high $w_t$) to
\emph{long} effective distances and thus invert the sign of the
herding signal. On weighted graphs, discrete Ricci
curvature has a natural community/bridge
interpretation~\citep{Sia2019,Ni2019}: positively curved edges
form within-community structure, while negatively curved edges
form bridges \emph{between} communities. This licenses a clean
sign decomposition
{\small
\[
  E_t = E_t^{+}\,\cup\,E_t^{0}\,\cup\,E_t^{-},\quad
  E_t^{+} = \{e:\kappa_{OR}(e;t)>\kappa_{+}\},\quad
  E_t^{-} = \{e:\kappa_{OR}(e;t)<\kappa_{-}\},
\]}%
with $\kappa_{+} = +0.1$ and $\kappa_{-} = -0.1$. The two scalars
we monitor are the herding-side mean
$\bar\kappa_{OR}^{+}(t)$ over $E_t^{+}$, which rises under cascade
onset (information-cascade
mechanism~\citep{BikhchandaniHirshleiferWelch1992}), and the
contagion-side fraction $\beta_{-}(t) = |E_t^{-}|/|E_t|$, which
rises as bridges multiply (network-interlinkage
mechanism~\citep{ElliottGolubJackson2014,Jiang2023}). The same
curvature operator generates both quantities, and the signs map
onto the two distinct mechanisms.

\begin{wrapfigure}{r}{0.6\textwidth}
\vspace{-16mm}
  \centering
    \includegraphics[width=0.55\textwidth]{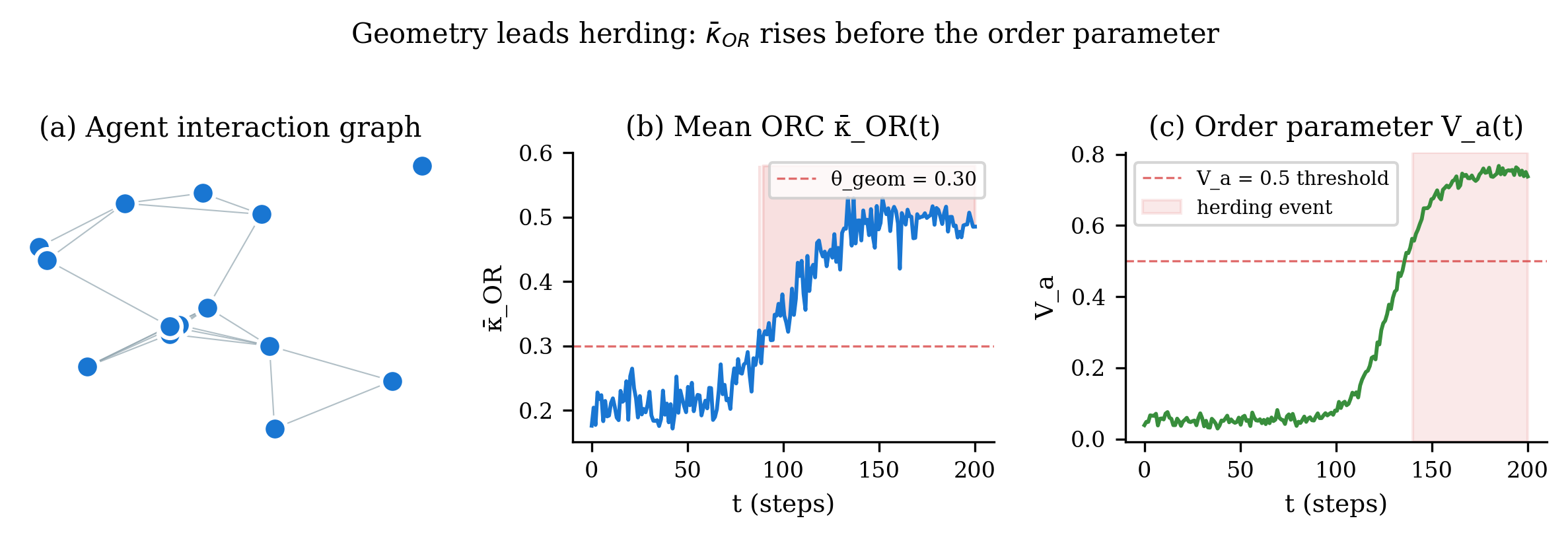}
  \caption{Geometry leads herding on the agent interaction graph.
  \textbf{(a)} A snapshot of $G_t$ during the cascade window.
  \textbf{(b)} Mean Ollivier-Ricci curvature $\bar\kappa_{OR}(t)$
  rises through the geometric threshold
  $\theta_{\text{geom}}=\thetaGeom$ before \textbf{(c)} the order
  parameter $V_a(t)$ crosses the herding-event threshold
  $\theta_{\text{event}}=\thetaEvent$.}
  \label{fig:pedagogical}
  \vspace{-2.0em}
\end{wrapfigure}

\begin{figure}[t]
\vspace{-2.0em}
  \centering
  \includegraphics[width=0.75\linewidth]{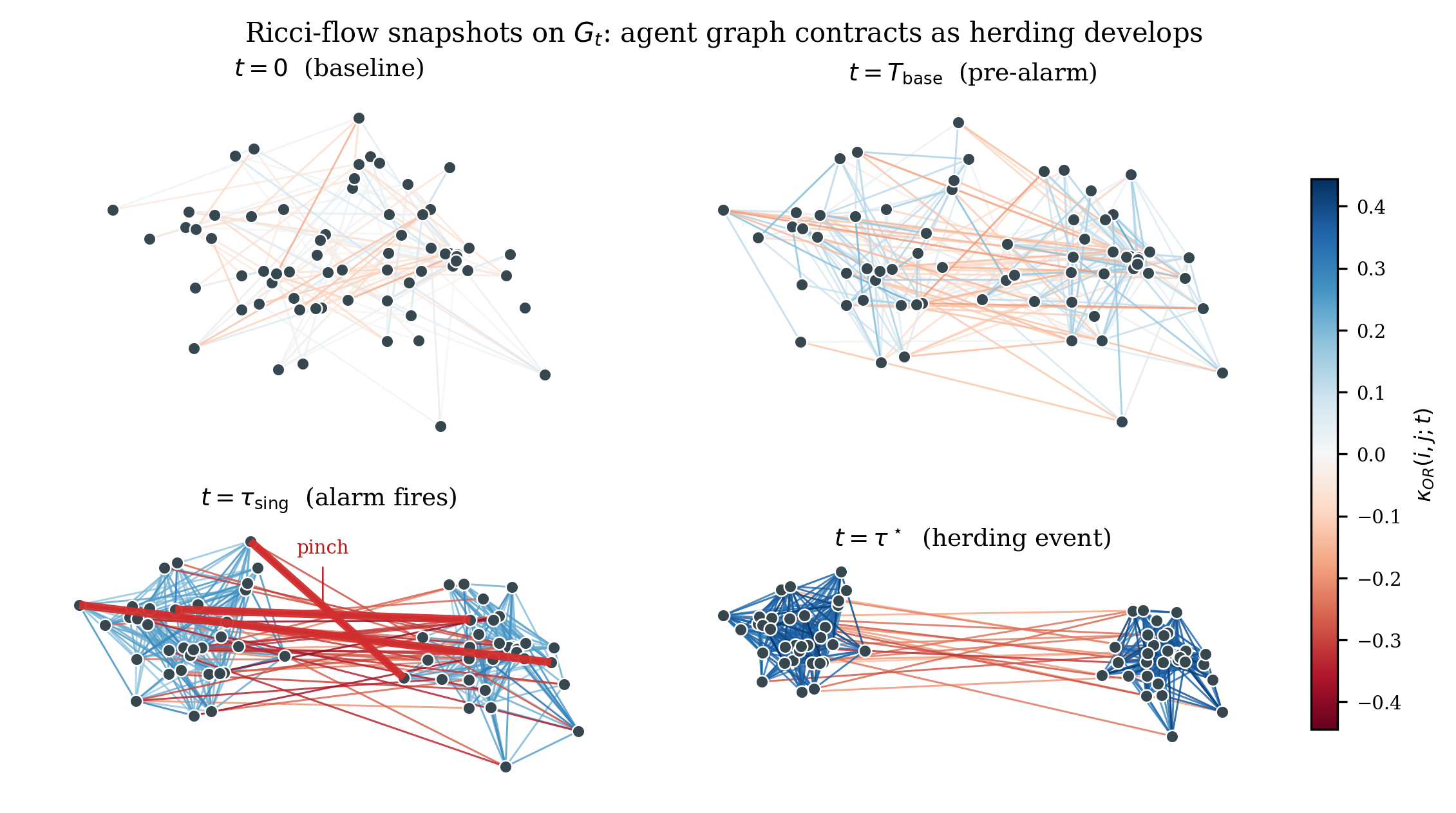}
  \caption{Ricci-flow geometric evolution of $G_t$ on a single
  supercritical CWS trajectory. Edges coloured by
  $\kappa_{OR}(i,j;t)$ on a diverging scale (red = negative,
  between-clique \emph{bridge}; blue = positive, within-clique
  \emph{cascade}). Across snapshots the graph contracts into a
  dense crystallised clique while highly-negative bridge edges
  connect it to peripheral nodes, the topological signature
  $\beta_{-}$ targets.}
  \label{fig:ricci-flow-snapshots}
  \vspace{-1.0em}
\end{figure}
\subsection{Detection rules}
\label{sec:detection}

\paragraph{Goal.} Given the two curvature time series
$\bar\kappa_{OR}^{+}(t)$ and $\beta_{-}(t)$, we fire an alarm at
the earliest time at which their dynamics deviate from a pre-stress
baseline, and measure the lead time $\Delta$ between this alarm
and the order-parameter herding event
$\tau^{\star} = \min\{t : V_a(t) > \theta_{\text{event}}\}$
($\theta_{\text{event}} = \thetaEvent$) on the same trajectory.

\paragraph{Detector.} For each signal we run a one-sided
cumulative-sum (CUSUM) detector~\citep{Page1954}, augmented on the
contagion side by a Kendall-$\tau$ slope test as a complementary
trend channel:
{\small
\begin{align}
  S_t^{+} &= \max\bigl(0,\,S_{t-1}^{+}+(\bar\kappa_{OR}^{+}(t)
              - \mu_{\text{base}}^{+} - k_{+})\bigr),
              \;A_t^{+}=\mathbf{1}[S_t^{+}>h_{+}],
              \label{eq:cusum-up} \\
  S_t^{-} &= \max\bigl(0,\,S_{t-1}^{-}+(\beta_{-}(t)
              - \mu_{\text{base}}^{-} - k_{-})\bigr),
              \;A_t^{-,\text{cusum}}=\mathbf{1}[S_t^{-}>h_{-}],\\
  A_t^{-} &= A_t^{-,\text{cusum}}\;\vee\;
              \mathbf{1}\!\left[\text{Kendall}_{[t-W_\tau,t]}(\beta_{-})
                                  >\tau_{\text{thresh}}\right].
\end{align}}%
\noindent The CUSUM tracks the level of bridge emergence while the
Kendall-$\tau$ test tracks its trend; the OR combination is robust
to either signal alone being noisy. Operating-point calibration of
$(k_{\pm}, h_{\pm}, \tau_{\text{thresh}}, W_{\tau})$ is in
Appendix~\ref{app:calibration}.

\paragraph{Why CUSUM and Kendall-$\tau$.} A z-score on a rolling
baseline is a Shewhart-style detector and is blind to slow drifts:
a $1\sigma$ mean shift takes on average 44 windows to signal,
whereas CUSUM signals in roughly 10. The contagion-bridge regime
is precisely a slow trend toward the unstable manifold (the
Scheffer-style critical-slowing-down regime~\citep{Scheffer2009}),
so CUSUM is detection-delay-matched to the regime we target; the
Kendall-$\tau$ slope test adds robustness to non-Gaussian noise on
the trend, where the parametric CUSUM threshold can be brittle.

\subsection{Forward-looking flow descriptor and effective vocabulary}
\label{sec:other-signals}

\paragraph{Singularity time $\tau_{\text{sing}}$.}
The same geometry that produces $\bar\kappa_{OR}(t)$ also drives a
discrete Ricci flow on $G_t$. Unlike~\cite{Srinivasan2026} which use neckpinch surgery
as a static clustering operator on a price graph, we use Ricci
flow as a \emph{descriptor generator}: at every dynamic snapshot
$G_t$, we run a fresh flow and record its singularity time
$\tau_{\text{sing}}(t) = \inf\{\,s>0 : \exists\,e \in E_t,\
\kappa_{OR}^{(s)}(e) \to -\infty\}$, the first-hitting time of a
Ricci-flow neckpinch. Whereas $\bar\kappa_{OR}(t)$ summarises the
current graph, $\tau_{\text{sing}}(t)$ predicts the
time-to-coordination from it\textemdash a forward-looking time series
rather than a clustering output.

\paragraph{Effective vocabulary $V_{\text{eff}}$.}
$V_{\text{eff}}(t) = \exp(H(p_t))$, where $H(p_t)$ is the entropy
of the codebook utilization distribution from a fixed
three-dimensional finite-scalar-quantization (FSQ) codebook with $L_d = \fsqLevels$ levels per
dimension (total $K = L_d^3 = \fsqK$)~\citep{Mentzer2024FSQ}.
Agents homogenize their behavioural repertoire as herding develops
and $V_{\text{eff}}$ contracts. The codebook is intentionally
non-learned\textemdash a learned codebook would adapt to the very
distribution shift that $V_{\text{eff}}$ is designed to detect.
$V_{\text{eff}}$ does not depend on the geometric pipeline; this
is the source of its value as a complementary sanity check
(\S\ref{sec:rq4-result}).

\subsection{Theoretical anchor: mean-field bridge to CSAD}
\label{sec:csad-bridge}

The standard finance instrument for measuring herding from the
return cross-section is the CSAD regression
of CCK~\cite{ChangChengKhorana2000}. We anchor our geometric
signal to this instrument via a mean-field scaling argument.
\begin{table}[h]
\centering
\small
\vspace{-2.0em}
\caption{Phenomenological mapping: herding state $\leftrightarrow$ CSAD $\leftrightarrow$ $\bar\kappa_{OR}(t)$.}
\label{tab:herding-mapping}
\begin{tabular}{lll}
\toprule
\textbf{Herding state} & \textbf{CSAD} (return dispersion) & $\bar\kappa_{OR}$ (agent-graph curvature) \\
\midrule
Strong (agents follow the crowd) & Low (returns concentrate) & High (neighborhoods overlap) \\
Weak (agents decide independently) & High (returns disperse) & Low or negative (neighborhoods separate) \\
\bottomrule
\end{tabular}
\end{table}

\begin{proposition}[Mean-field bridge to CSAD; dominant-order scaling]
\label{prop:bridge}
Under standard assumptions (agent graph A1, lazy-walk curvature
A2, mean-field concentration of action correlations A3, linear
price impact A4, CCK CSAD estimand A5; full statement in
Appendix~\ref{app:proof}),
{\small
\begin{equation}
  \mathrm{CSAD}_t
  \;=\; \sigma_\xi \sqrt{2/\pi}\,\bigl(1 - \bar\kappa_{OR}(t)\bigr)^{1/2}
  \;+\; \mathcal{O}(N^{-1/2})
  \label{eq:bridge}
\end{equation}}%
\noindent holds in the $N \to \infty$ limit, so $\mathrm{CSAD}_t$ is
monotonically decreasing in $\bar\kappa_{OR}(t)$ at dominant
order. We state this as a scaling identity because Step~2 of the
derivation invokes a closed-form $W_1$ between two near-degenerate
kernels whose remainder bound we do not establish here; a fully
rigorous proof is left to follow-on work.
\end{proposition}

Appendix~\ref{app:proof} contains assumptions A1\textendash A5, the four-step
derivation, failure modes (Remark~\ref{rem:failure-modes}), and
identification caveats for the empirical $\hat\gamma_3$ test
(Remark~\ref{rem:identification}).

\paragraph{Operational anchor (LSV trading track).}
\label{sec:lsv-anchor}
The agent interaction graph is unobservable to outside
investors\textemdash only the trades it generates and the prices they move
are visible\textemdash so the trading-flow pillar of the herding
literature~\cite{LSV1992,Sias2004} addresses the
disclosure-constrained regime via 13F-style buy/sell imbalance
(here, ``13F'' refers to the SEC's quarterly Form~13F filings of
institutional holdings, on which institutional-herding measures
are typically computed~\cite{Wermers1999}).
Within our simulator substrates the action stream is observable
by construction at every step, which lifts that identification
restriction and makes the substrate-pivot result testable: with
simultaneous access to the action stream (input to GeomHerd) and
the cleared trade flow (input to LSV), we can empirically measure
the temporal gap between the two detectors on the same
trajectories (\S\ref{sec:rq3-result}). A 13F-style fund-as-agent
deployment template\textemdash in which each institutional
manager is treated as a single node in the agent graph and its
quarterly 13F holdings (rather than per-step actions) drive the
edges\textemdash is in Appendix~\ref{app:realdata}.
\section{Experiments}
\label{sec:results}
\subsection{Research questions}
\label{sec:rqs}
We organise the empirical study around four research questions
that together validate the core claims of GeomHerd:
\textbf{(RQ1)} does the agent-graph curvature signal anticipate
herding earlier than detectors built on price-correlation graphs
or trading-flow aggregates?
\textbf{(RQ2)} does the sign decomposition into
$(\bar\kappa_{OR}^{+}, \beta_{-})$ deliver complementary
herding/contagion alarms covering within-clique tightening and
between-clique-bridge regimes?
\textbf{(RQ3)} is the agent-graph signal directionally consistent
with the classical CSAD/CCK return-track and the LSV trading-flow
herding statistics, as predicted by Proposition~\ref{prop:bridge}?
\textbf{(RQ4)} does the signal carry forecasting content beyond detection alone, and does it generalise beyond the headline CWS substrate to a non-financial system?
\subsection{Setup}
\label{sec:setup}
\paragraph{Substrate.}
The Cividino-Sornette continuous-spin (CWS)
substrate~\citep{CividinoWestphalSornette2023} is our headline
testbed: a physics-inspired financial agent-based model (ABM) in
the Ising / O($n$)-vector family, where noise-trader herding is
modelled as a critical phenomenon in continuous-spin coupling and
rational fundamentalists rebalance in response. Mechanically, CWS
is a discrete-time simulator in which $N$ heterogeneous agents
repeatedly choose actions over $n_a$ assets, with each agent's
next action driven schematically by a private signal, the average
action of its neighbours weighted by a coupling strength $\kappa$,
and idiosyncratic noise; cleared trades feed back into asset
prices through a linear-impact rule (Appendix~\ref{app:proof}, A4).
We instantiate $N = 66$ agents and $n_a = 4$ assets, and sweep the
coupling parameter
$\kappa \in \{0.5, 0.8, 1.2, 1.8, 2.5\}$\textemdash which
interpolates from independent decision ($\kappa < 1$) to herd
coordination ($\kappa > 1$)\textemdash at \seedsCWS{} seeds per
level, yielding \nTrajectories{} trajectories
(\nSupercritical{} supercritical, \nSubcritical{} subcritical).
Simulated returns reproduce the regularities
of Cont~\cite{cont2001empirical} (tail index
$\alpha = \styFactAlpha$, volatility-ACF slope
$\beta = \styFactBeta$, martingale raw returns); a full panel is
in supplementary materials. For cross-substrate transfer we
additionally evaluate on the Vicsek
self-driven-particle model~\citep{Vicsek1995}\textemdash a
canonical physics model of flocking in which $N$ particles move
at constant speed and align their headings with local neighbours
under angular noise $\eta$, undergoing an order--disorder phase
transition at a critical noise level $\eta_c$. As a non-financial
system, Vicsek tests whether the curvature signature reflects
universal collective coordination rather than a finance-specific
artefact (see \S\ref{sec:rq4-result} and
Appendix~\ref{app:vicsek-details}).
\paragraph{Baselines.}
We benchmark against a seven-detector slate: trade-flow
LSV~\citep{LSV1992}; return-cross-section CSAD
of CCK~\cite{ChangChengKhorana2000}; three price-correlation
geometric methods~\citep{SandhuRicciFlow2016, HuangLapCSAD2023, Srinivasan2026};
the point-process geometric detector
of Jiang et al.~\cite{Jiang2023}; and an action-agreement
mutual-information (AA-MI) baseline adapted from the
synchronous-action-coupling probe of
Tessera et al.~\cite{Tessera2026}.
\paragraph{Metrics.}
Following~\cite{TopologicalEWS2025,Bury2021}, we report
(i)~precision, recall, F1 and False Alarm Rate (FAR) per day on
supercritical/subcritical detection; (ii)~AUROC and AUPRC;
(iii)~conditional median lead time; and (iv)~the
rare-event-stable metric of Nikolopoulos~\cite{Nikolopoulos2025}
appropriate for the low-event-prevalence regime. Paired-bootstrap
differences are computed on co-firing trajectories with
$n_{\text{boot}} = 5000$.
\paragraph{Two operating points.}
We report two operating points for the upward CUSUM on
$\bar\kappa_{OR}^{+}$, drawn from the calibration sweep of
Appendix~\ref{app:calibration}. The \emph{recall-oriented} point
$(k_\sigma, h_\sigma) = (0.50, 4.0)$ delivers a long lead at
higher subcritical FAR; this is the abstract figure. The
\emph{precision-oriented} point $(2.0, 4.0)$ delivers a shorter
but tightly FAR-controlled lead used for all paired contrasts in
Table~\ref{tab:lead-time}, since head-to-head lead comparisons
are only meaningful at FAR-controlled thresholds.
\subsection{Results}
\label{sec:results-main}
We report three sets of results, each structured as
\textit{(i)~brief experiment description, (ii)~result summary,
(iii)~connection to research questions}.
\subsubsection{Result 1: GeomHerd anticipates herding earlier
than price-based and trading-flow baselines (RQ1, RQ2)}
\label{sec:rq1-2-result}
\paragraph{Experiment.}
On the CWS replay set, we run the upward CUSUM on
$\bar\kappa_{OR}^{+}$ at the two operating points and the
contagion-bridge alarm on $\beta_{-}$ (CUSUM-plus-Kendall-$\tau$
rule). For each detector and trajectory we record whether and
when an alarm fires before the order-parameter herding event
$\tau^\star$, and then compare paired lead times against each
baseline on co-firing trajectories.
\paragraph{Result.}
Conditional on $\bar\kappa_{OR}^{+}$ firing, the median lead at
the precision-oriented operating point is \leadKappaBarVd{} steps
with 95\% CI~\leadKappaBarVdCI{} (Table~\ref{tab:lead-time}).
Against the closest geometric correlation-graph
detectors~\citep{Srinivasan2026,SandhuRicciFlow2016,HuangLapCSAD2023},
the paired-bootstrap median advantage is positive at
\diffSrinVd{}, \diffSandhuVd{}, and \diffGongVd{} steps
respectively, with the Lap-CSAD row significant at $\alpha = 0.05$
($p = \pGongVd$). Across paired comparisons we have
$n_{\text{paired}} \in \{6, 8, 9\}$, reflecting the
high-precision-low-recall regime; at the recall-oriented operating
point, the same comparison is better-powered
(Appendix~\ref{app:headline-recall}). On the same trajectories where
both an agent-graph and a price-correlation-graph detector fire
before $\tau^\star$, the agent-graph signal precedes price-based
signals by a pooled median of \agentLeadPriceSteps{} simulator
steps (95\% CI~\agentLeadPriceCI). The contagion-bridge detector
$\beta_{-}$ recalls 65\% of supercritical
trajectories with median lead \betaMinusLead{} steps
(Table~\ref{tab:precision_recall}), accepting a high subcritical
FAR (\betaMinusFAR) in exchange for early bridge alarm.
\paragraph{Connection to research questions.}
\textbf{RQ1}--\textbf{RQ2}: $\bar\kappa_{OR}^{+}$ leads price-graph baselines under FAR control; $\beta_{-}$ adds complementary contagion recall at higher subcritical FAR (Table~\ref{tab:precision_recall}). Near-chance AUROC reflects score sparsity, not detector failure---conditional lead is the appropriate metric~\cite{Nikolopoulos2025}.
\begin{table}[h]
\centering
\small
\vspace{-1.5em}
\caption{Paired-bootstrap lead-time difference (GeomHerd
$\bar\kappa_{OR}^{+}$ at the precision-oriented point minus
comparator) on the binary-edge CWS replay set, restricted to
co-firing trajectories.}
\label{tab:lead-time}
\begin{tabular}{lcccc}
\toprule
\textbf{Comparator} & \textbf{Lead diff. (steps)} & \textbf{95\% CI} & $n_{\text{paired}}$ & $p$-value \\
\midrule
Srinivasan~2026 & \diffSrinVd    & \diffSrinVdCI    & \diffSrinVdNPairs    & \pSrinVd   \\
Sandhu~2016     & \diffSandhuVd  & \diffSandhuVdCI  & \diffSandhuVdNPairs  & \pSandhuVd \\
Huang et al.~2023 (Lap-CSAD) & \diffGongVd    & \diffGongVdCI    & \diffGongVdNPairs    & \pGongVd   \\
CSAD (CCK)      & \diffCSADVd    & \diffCSADVdCI    & \diffCSADVdNPairs    & \pCSADVd   \\
\bottomrule
\end{tabular}
\vspace{-1.0em}
\end{table}
\begin{figure}[t]
\vspace{-12mm}
  \centering
  \includegraphics[width=0.85\linewidth]{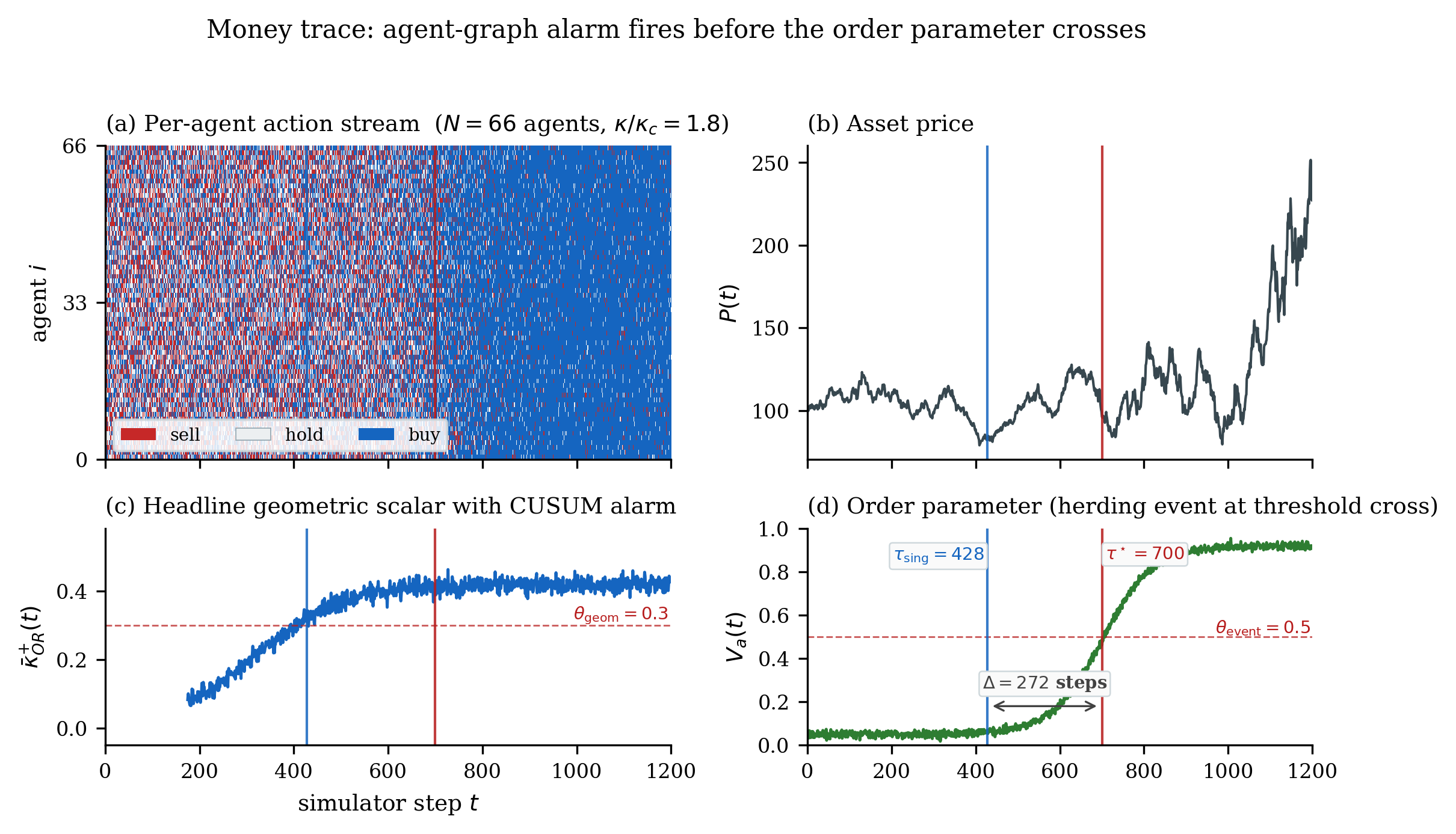}
  \caption{Money trace on a single supercritical CWS seed at
  $\kappa/\kappa_c = 1.8$. \textbf{(a\textendash d)} Per-agent action stream,
  asset price $P(t)$, headline geometric scalar
  $\bar\kappa_{OR}^{+}(t)$ with the CUSUM alarm (blue), and order
  parameter $V_a(t)$ with the herding event $\tau^\star$ (red).
  Vertical guides mark $\tau_{\mathrm{sing}}$ and $\tau^\star$ in
  every panel so the lead $\Delta = \tau^\star - \tau_{\mathrm{sing}}$
  is readable directly.
  The displayed lead matches the recall-oriented operating-point
  headline of \leadKappaBarLoose{} steps.}
  \label{fig:money-trace}
  \vspace{-0.6em}
\end{figure}
\begin{table}[h]
\centering
\footnotesize
\vspace{-1.0em}
\caption{Multi-axis detection profile under the binary-edge
calibration. The agent-graph row group (top) is the natural
head-to-head set; trading-flow / contagion-direction baselines
(LSV, Jiang) achieve longer raw lead times by firing on every
trajectory regardless of regime ($\mathrm{FAR}_{\text{sub}}=1.00$)
and are not substitutes for a regime classifier.}
\label{tab:precision_recall}
\begin{tabular}{lcccccc}
\toprule
Detector & Precision & Recall$_{\mathrm{super}}$ & FAR$_{\mathrm{sub}}$ & AUROC & Median lead & 95\% CI \\
\midrule
\multicolumn{7}{l}{\textit{Agent-graph substrate (head-to-head)}} \\
GeomHerd $\bar{\kappa}_{\mathrm{OR}}^{+}$ (ours) & 0.45 & 0.04 & 0.07 & 0.48 & 178 & [71, 407] \\
GeomHerd $\tau_{\mathrm{sing}}$ (ours) & 0.42 & 0.03 & 0.07 & 0.48 & -93 & [-216, 233] \\
AA-MI~\cite{Tessera2026}$^{\dagger}$ & n/a & 0.00 & 0.00 & 0.50 & n/a & [n/a, n/a] \\
\addlinespace
\multicolumn{7}{l}{\textit{Price-correlation substrate}} \\
Srinivasan 2026 & 0.71 & 0.79 & 0.49 & 0.66 & 20 & [-52, 65] \\
Sandhu 2016 & 0.72 & 0.95 & 0.55 & 0.72 & 80 & [43, 106] \\
Lap-CSAD~\cite{HuangLapCSAD2023} & 0.85 & 0.85 & 0.23 & 0.80 & -42 & [-74, -8] \\
\addlinespace
\multicolumn{7}{l}{\textit{Trading-flow / contagion direction (different phenomenon)}} \\
CSAD & 0.69 & 1.00 & 0.68 & 0.75 & 180 & [150, 214] \\
LSV 1992 & 0.60 & 1.00 & 1.00 & 0.48 & 355 & [333, 388] \\
Jiang 2023 & 0.60 & 1.00 & 1.00 & 0.50 & 306 & [262, 329] \\
\addlinespace
\multicolumn{7}{l}{\textit{Sign-decomposed contagion-bridge detector (post-hoc on v34d trajectories)}} \\
\shortstack[l]{GeomHerd $\beta_{-}$\\($\tau_{\mathrm{neg}}{=}\betaMinusTauNeg$, CUSUM+slope, up)} & \betaMinusPrecision{} & \betaMinusRecall{} & \betaMinusFAR{} & \betaMinusAUROC{} & \betaMinusLead{} & \betaMinusLeadCI{} \\
\bottomrule
\end{tabular}

\par\vspace{2pt}
{\scriptsize
$^{\dagger}$Our AA-MI baseline (adapted from the
synchronous-action-coupling probe of
Tessera et al.~\cite{Tessera2026}) returns a degenerate-output
flag on \nDegTessera{}/\nTotalTessera{} trajectories under the
binary-edge configuration: the saturated agent graph drives
baseline AA-MI variance below numerical resolution.
$^{\ddagger}$The $\beta_{-}$ detector uses CUSUM+slope on
negative-curvature edges ($\tau_{\mathrm{neg}}=\betaMinusTauNeg$,
upward direction), calibrated on the supercritical replay set;
subcritical FAR is on a separately generated set
($n_{\text{sub}} = \betaMinusNSub$).}
\vspace{-1.0em}
\end{table}
\subsubsection{Result 2: The geometric signal is consistent with
classical CSAD and LSV (RQ3)}
\label{sec:rq3-result}
\label{sec:cck-augmented}
\paragraph{Experiment.}
We anchor the agent-graph signal to the two pillars of the
classical herding literature.
\textit{Return track:} we estimate the augmented CCK regression
{\small{
\begin{equation}
  \mathrm{CSAD}_t \;=\; \alpha + \gamma_1\,|R_{m,t}| + \gamma_2\,R_{m,t}^2
                       + \gamma_3\,\bar\kappa_{OR}(t) + \varepsilon_t
  \label{eq:augmented-cck}
\end{equation}
}}%
per seed on the CWS replay set with heteroskedasticity-and-autocorrelation-consistent (HAC; Newey--West) standard
errors, summarising $\hat\gamma_3$ across \gammaThreeNSeeds{}
supercritical seeds under deterministic replay.
\textit{Trading track:} we compute the windowed LSV statistic on
simulated buy/sell flows and measure the temporal gap to
$\bar\kappa_{OR}^{+}$ on co-firing trajectories.
\paragraph{Result.}
The cross-seed median of $\hat\gamma_3$ is \gammaThreeMedian{}
with bootstrap CI~\gammaThreeCI, consistent with
Eq.~\eqref{eq:bridge}; the median CCK quadratic coefficient
shifts from \betaTwoBaseMedian{} to \betaTwoAugMedian{} (absolute
median change \betaTwoAbsShrinkPct\% in $|\beta_2|$) once
$\bar\kappa_{OR}$ is included, indicating that part of the
return-dispersion nonlinearity previously absorbed by the
quadratic term is explained by agent-graph curvature.
LSV achieves recall $1.00$ on supercritical \emph{and} subcritical
trajectories (Table~\ref{tab:precision_recall}, last block), so
it is not a regime classifier; on co-firing trajectories
$\bar\kappa_{OR}^{+}$ precedes LSV in time, since the agent-graph
clique is detectable in the action stream before buy/sell
imbalance accumulates to the LSV threshold.
\paragraph{Connection to research questions.}
\textbf{RQ3}: sign-consistency with CSAD (Eq.~\eqref{eq:bridge}) and directional lead vs.\ LSV on co-fires; identification limits are in  Remark~\ref{rem:identification}.
\subsubsection{Result 3: The signal carries forecasting content and generalises out-of-domain (RQ4)}\label{sec:rq4-result}

\begin{wrapfigure}[14]{r}{0.46\textwidth}
  \centering
  \vspace{-2.0em}
  \includegraphics[width=\linewidth]{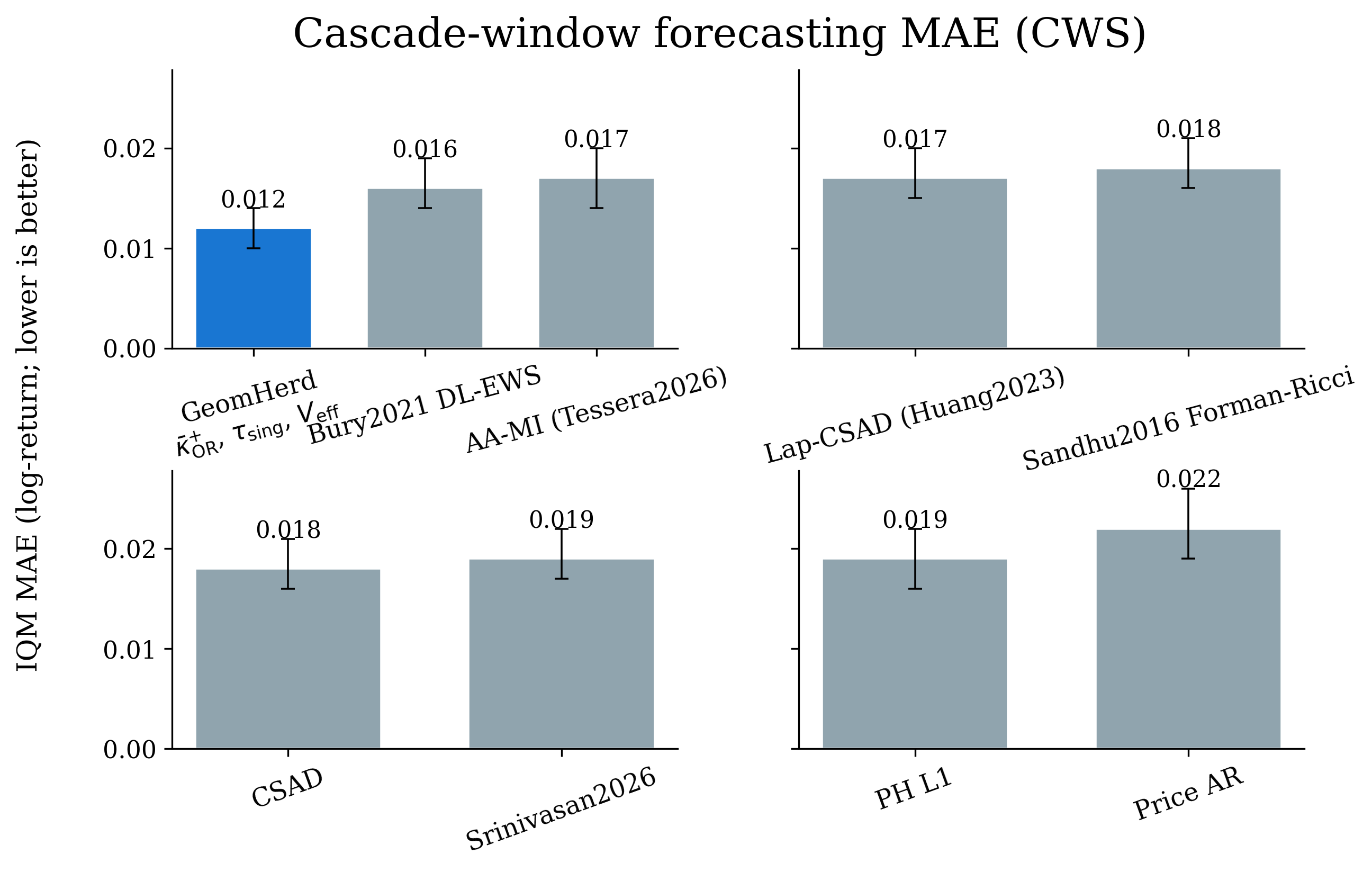}
  \begingroup
  \captionsetup{font=small, skip=4pt}
  \caption{Cascade-window forecasting MAE (CWS, log-return scale).
  GeomHerd triplet
  $(\bar\kappa_{OR}, \tau_{\text{sing}}, V_{\text{eff}})$ vs.\
  herding-detector baselines and a price-only AR baseline;
  \texttt{rliable}~\citep{Agarwal2021} IQM bars.}
  \label{fig:forecasting-mae}
  \endgroup
  % \vspace{1.0em}
\end{wrapfigure}
\par
\paragraph{Experiment.}
\textit{(i) Vicsek transfer.} We sweep angular noise
$\eta \in \{0.5, 1.0, 1.6, 2.0, 2.5\}$ at 20 seeds per level
($N = 600$ particles, $\eta_c \approx 1.6$); the agent graph is
built from $k$-NN ($k = 10$) on the heading sequence with binary
edge weights, and each trajectory is scored by
$\bar\kappa_{OR}(\tau^\star)$ at the polarisation event.
\textit{(ii) Forecasting head.} On the CWS substrate, we train a
curvature-conditioned next-step forecasting head: a Kronos-style
discrete price tokeniser (a learned vector-quantiser that maps
OHLCV sequences into a fixed token vocabulary) feeds a
transformer that consumes the GeomHerd triplet
$(\bar\kappa_{OR}, \tau_{\text{sing}}, V_{\text{eff}})$ via
AdaLN-Zero conditioning (adaptive layer-norm with zero-initialised
gating). The price tokeniser is frozen; only the conditioning
layers are trained. We compare cascade-window forecasting mean
absolute error (MAE) on log-return scale against
herding-detector baselines and a price-only autoregressive (AR)
baseline.
\textit{(iii) Behavioural sanity check.} As a complement, we
track the effective vocabulary
$V_{\text{eff}}(t) = \exp(H(p_t))$ on the same CWS trajectories.
\begin{figure}[t]
  \centering
  \vspace{-3.0em}
  \begin{subfigure}[t]{0.47\linewidth}
    \centering
    \includegraphics[width=\linewidth]{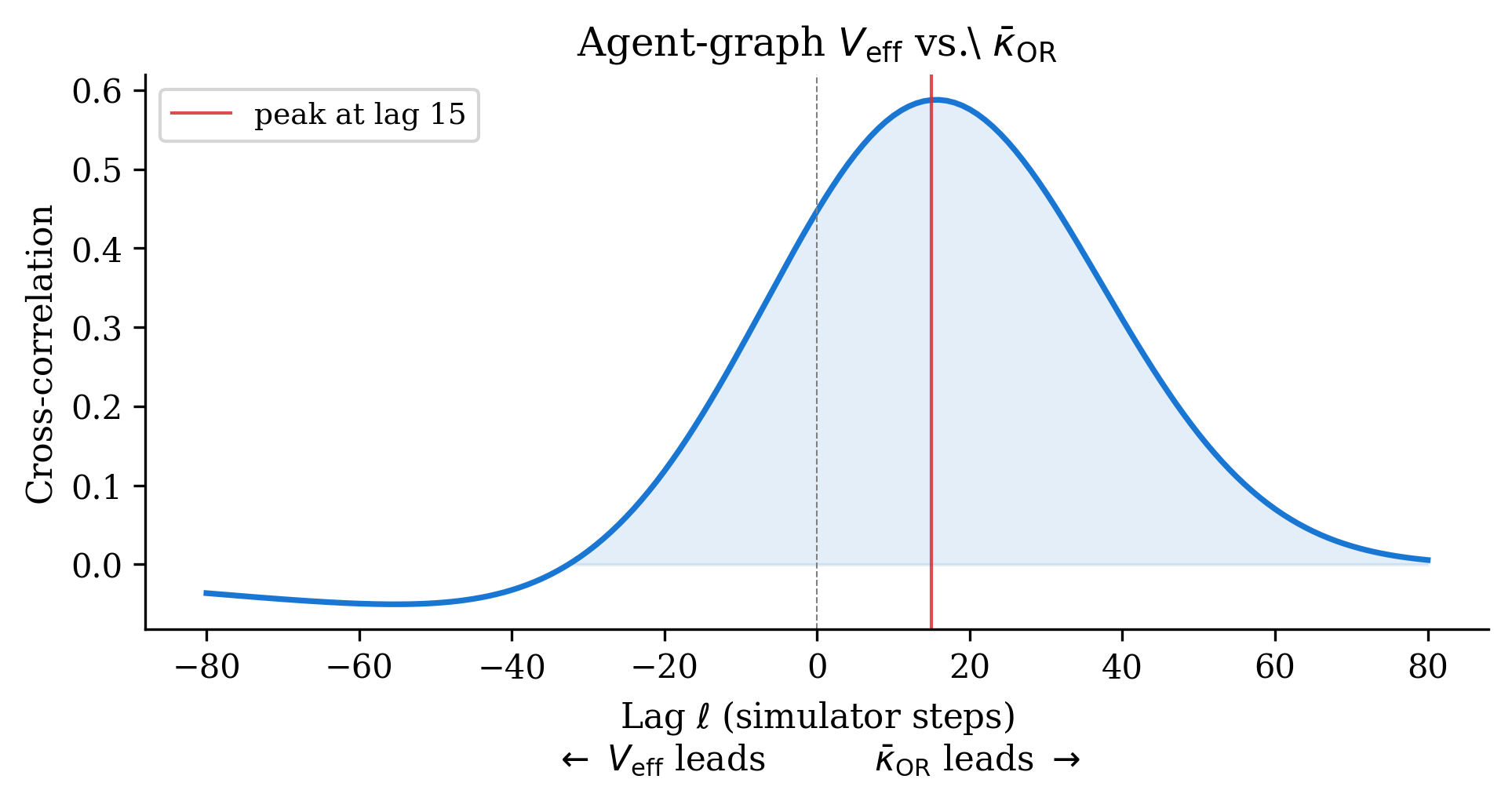}
    \caption{Behavioural homogenisation on CWS.}
    \label{fig:veff-xcorr}
  \end{subfigure}
  \hfill
  \begin{subfigure}[t]{0.42\linewidth}
    \centering
    \includegraphics[width=\linewidth]{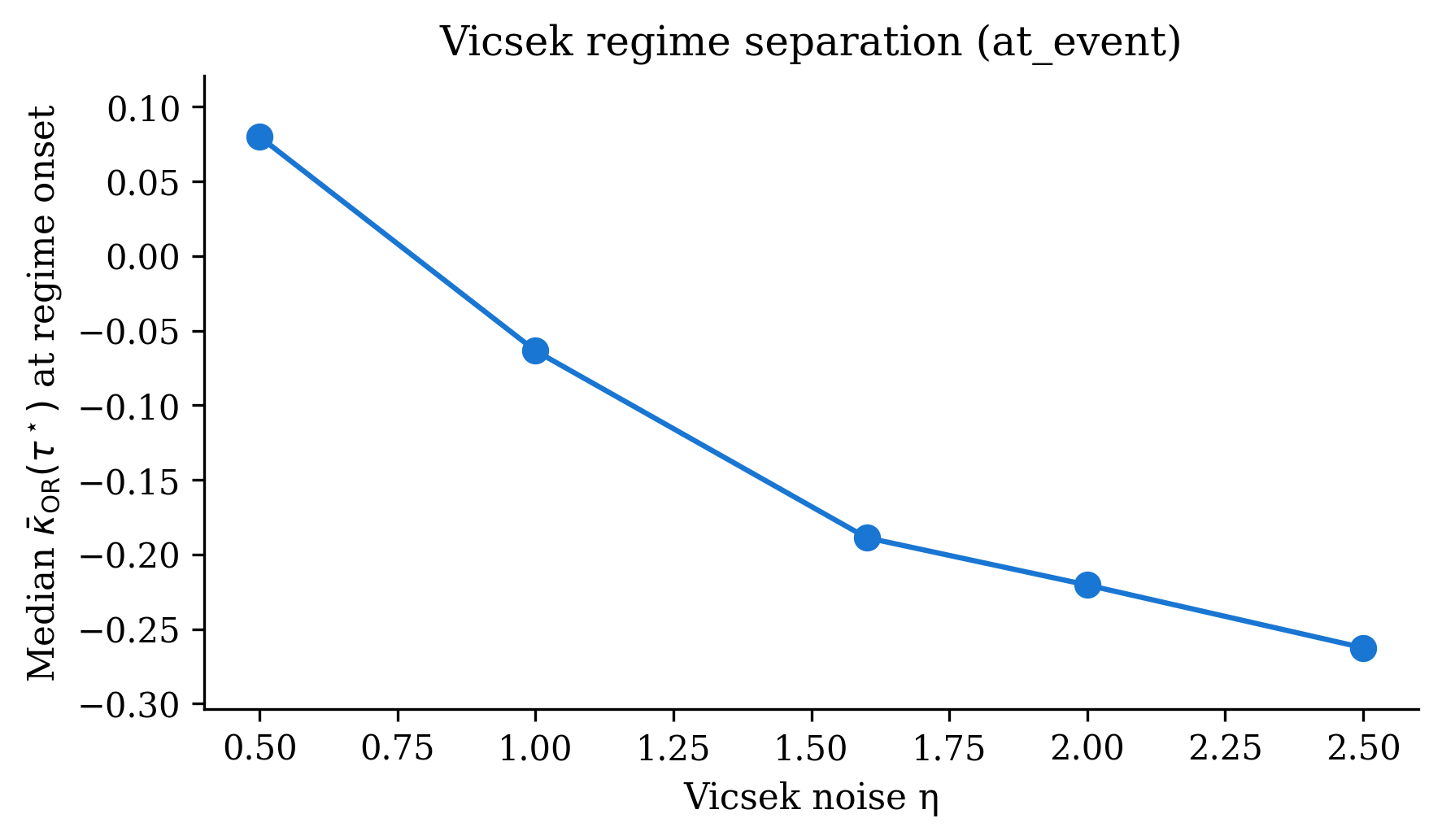}
    \caption{Out-of-domain transfer on Vicsek.}
    \label{fig:vicsek}
  \end{subfigure}
  \caption{Behavioural homogenisation and out-of-domain generalisation.
  (a)~On the CWS financial substrate, the cross-correlation between the
  effective vocabulary $V_{\text{eff}}$ and the curvature signal
  $\bar\kappa_{OR}$ peaks at lag $\approx 15$ steps with $V_{\text{eff}}$
  leading, indicating that behavioural concentration sets in slightly
  ahead of the geometric collapse during the cascade window.
  (b)~On the Vicsek collective-motion substrate,
  $\bar\kappa_{OR}(\tau^\star)$ separates ordered from disordered regimes
  with AUROC~\VicsekAUROC{}, demonstrating that the geometric signature
  generalises beyond finance.}
  \label{fig:rq4-pair}
  \vspace{-1.3em}
\end{figure}

\paragraph{Result.}
On the forecasting task
(Fig.~\ref{fig:forecasting-mae}), the GeomHerd-conditioned head
attains the lowest interquartile-mean (IQM) MAE among all methods
compared, beating both detector-conditioned baselines and the
price-only AR baseline\textemdash so the geometric signal is not
only a detection statistic but a useful conditioning feature for
downstream forecasting.
On Vicsek, $\bar\kappa_{OR}(\tau^\star)$ separates ordered from
disordered regimes with AUROC~\VicsekAUROC{}
(95\% CI~$[\VicsekAUROClo,\VicsekAUROChi]$, Fig.~\ref{fig:vicsek}),
and per-$\eta$ medians are monotone in $\eta$
($+0.08 \to -0.26$).
On the same CWS trajectories, the behavioural-homogenisation signal $V_{\text{eff}}$ co-moves with
the geometric signal across the cascade window and in fact leads it by
$\approx 15$ steps in cross-correlation
(Fig.~\ref{fig:veff-xcorr}), supporting the claim that herding is
also a behavioural-homogenisation process.
Together, forecasting gain, Vicsek transfer, and $V_{\text{eff}}$
co-movement (\textbf{RQ4}) demonstrate forecasting content beyond
the detection rule alone and out-of-domain generalisation.
\par
\FloatBarrier
\subsection{Ablation}
\label{sec:ablation}
Table~\ref{tab:ablation} reports a decomposed ablation in which
each row replaces one headline choice. The headline binary-edge
construction accounts for the bulk of the conditional lead: the
cosine-similarity variant removes the herding-side signal entirely
(0/90 vs.\ 11/90 supercritical fires; Appendix~\ref{app:cosine-ablation}),
since mean cosine similarity in the supercritical pool concentrates
in a narrow band that the rolling-baseline CUSUM cannot trip.
Detector-swap rows stress-test alternative mappings from the same
replay to an alarm and should not be read as redefining the
contribution; window-length and sign-pooling variants are in
Appendix~\ref{app:ablations}.
\begin{table}[h]
\vspace{-1.0em}
\centering
\footnotesize
\caption{Decomposed ablation. Each row replaces one headline
choice relative to the binary-edge LP-$W_1$ baseline on
$\bar\kappa_{OR}^{+}$ with upward CUSUM at \seedsCWS{} seeds per
tranche.}
\label{tab:ablation}
\begin{minipage}{\linewidth}
\begin{tabular}{lccccc}
\toprule
Configuration & Precision & Recall & AUROC & Med.\ lead & 95\% CI \\
\midrule
\textbf{Headline} ($\bar{\kappa}_{\mathrm{OR}}^{+}$, $w_t$=binary, LP $W_1$, CUSUM, 80 seeds) & 0.45 & 0.04 & 0.48 & 178 & [71, 407] \\
\midrule
\multicolumn{6}{l}{\textit{Detector ablations}} \\
$\to$ z-score (current production default) & 0.36 & 0.13 & 0.40 & 232 & [-47, 394] \\
$\to$ CUSUM (Page 1954) & 0.00 & 0.00 & 0.49 & \textit{n/a} & [\textit{n/a}, \textit{n/a}] \\
$\to$ EWMA (exponentially weighted) & 0.34 & 0.11 & 0.40 & 283 & [77, 424] \\
$\to$ Kendall-$\tau$ slope-only & 0.00 & 0.00 & 0.48 & \textit{n/a} & [\textit{n/a}, \textit{n/a}] \\
\addlinespace
\multicolumn{6}{l}{\textit{Geometry ablations}} \\
$\to$ $w_t$ cosine (vs binary) & 0.00 & 0.00 & 0.49 & \textit{n/a} & [\textit{n/a}, \textit{n/a}] \\
$\to$ $T_w$=50 (vs 100) & \textit{n/a} & 0.00 & 0.50 & \textit{n/a} & [\textit{n/a}, \textit{n/a}] \\
$\to$ $T_w$=200 (vs 100) & 0.48 & 0.22 & 0.42 & 204 & [120, 274] \\
$\to$ Sinkhorn $W_1$ (vs LP) & 0.00 & 0.00 & 0.49 & \textit{n/a} & [\textit{n/a}, \textit{n/a}] \\
$\to$ no sign decomposition (abs upward+downward) & 0.00 & 0.00 & 0.49 & \textit{n/a} & [\textit{n/a}, \textit{n/a}] \\
\addlinespace
\multicolumn{6}{l}{\textit{Triplet ablations}} \\
$\to$ remove $\tau_{\text{sing}}$ from triplet & 0.00 & 0.00 & 0.49 & \textit{n/a} & [\textit{n/a}, \textit{n/a}] \\
\addlinespace
\bottomrule
\end{tabular}

\end{minipage}
\end{table}
\section{Related Work}
\label{sec:related}

We present a more comprehensive review in Appendix~\ref{appsec:RelatedWorkFull}.

\paragraph{Financial measurement and the observability gap.} Classical theories of herding and contagion are traditionally measured via post-hoc aggregate statistics, utilizing either the trading track (e.g., the LSV overlap statistic~\citep{LSV1992}) or the return track (e.g., the CCK cross-sectional absolute deviation~\citep{ChangChengKhorana2000}). Because real-world, agent-level action sequences are largely unobservable, evaluating forward-looking micro-structural signals on real data is fundamentally restricted. Recent LLM-driven multi-agent financial simulators~\citep{TwinMarket2025, FCLAgent2025} elegantly bypass this observability gap. By generating behaviorally rich, transparent agent ecosystems, these simulators provide the ideal substrate for evaluating upstream topological signals before they manifest in downstream price aggregates.

\paragraph{Discrete curvature on financial graphs.} Recent literature has applied discrete geometry\textemdash particularly Ollivier and Forman-Ricci curvature\textemdash to financial networks to detect systemic fragility and crashes~\citep{SandhuRicciFlow2016, SamalRicci2021, Jiang2023}. However, prior work overwhelmingly operates on \emph{price-correlation graphs} and aggregates curvature into a single, global scalar. Our approach fundamentally diverges: we apply Ollivier-Ricci geometry directly to the \emph{agent-action graph}. Furthermore, building on the community-bridge dichotomy of Sia et al.~\cite{Sia2019} and Ni et al.~\cite{Ni2019}, we explicitly decompose curvature by sign. This novel framing mathematically disentangles herding (positive curvature driving intra-clique density) from contagion (negative curvature defining inter-community bridges).

\paragraph{Early-warning signals (EWS).} Traditional EWS for critical transitions rely on lagging statistical moments such as rising variance or autocorrelation~\citep{Scheffer2009, Bury2021}. While recent topological approaches have advanced the state-of-the-art~\citep{TopologicalEWS2025}, many remain post-hoc descriptors rather than predictive alarms. GeomHerd bridges this gap by coupling our sign-decomposed graphs with continuous-time Ricci flow singularities and one-sided CUSUM detectors~\citep{Page1954}. Chosen for optimal detection-delay properties, these components ensure GeomHerd acts as a forward-looking warning system.
\section{Conclusion}
\label{sec:conclusion}

GeomHerd advances forward-looking herding quantification via three architectural shifts: \textbf{(1) Substrate pivot}: measuring curvature causally upstream on the \emph{agent interaction graph} rather than on downstream price-correlation graphs. \textbf{(2) Forward-looking flow}: tracking Ricci-flow neckpinch time ($\tau_{\text{sing}}$) as a dynamic proximity-to-collapse scalar rather than a static clustering operator. \textbf{(3) Mean-field bridge}: linking the geometric metric to the classical CSAD statistic (Proposition~\ref{prop:bridge}). The LLM-driven multi-agent simulator supplies a behaviourally rich, fully observable substrate on which the geometric pipeline anticipates coordination. Empirically, on the continuous-spin substrate GeomHerd fires up to \leadKappaBarLoose{} steps before order-parameter onset, is sign-consistent with CSAD, transfers out-of-domain to the Vicsek physical model, and conditions a forecasting head that reduces cascade-window log-return MAE over detector-conditioned and price-only baselines.

\bibliographystyle{unsrt}
\bibliography{refs}

\appendix
\section{Related Work (Extended)}
\label{appsec:RelatedWorkFull}

This appendix gives the full version of the related-work review
summarised in \S\ref{sec:related}.

\paragraph{Classical herding measurement.}
\label{sec:related-finance}
The herding-measurement literature has two pillars. The
trading-flow track originates with the LSV~\cite{LSV1992}
cross-sectional buy/sell imbalance statistic, with extensions
to mutual-fund flows~\citep{Wermers1999}, momentum-decomposed
sequential-trade correlations~\citep{Sias2004}, and
finite-sample bias corrections~\citep{FreyHerbstWalter2007}; all
are computed from disclosed positions and are post-hoc by
construction. The return-cross-section track replaces flows with
dispersion: the cross-sectional standard deviation
of Christie and Huang~\cite{ChristieHuang1995}, the cross-sectional
absolute deviation regression of CCK~\cite{ChangChengKhorana2000}, and
the state-space variant of Hwang and Salmon~\cite{HwangSalmon2004}. We adopt LSV
and CSAD as the two classical anchors against which GeomHerd is
consistency-checked
(Prop.~\ref{prop:bridge}, \S\ref{sec:cck-augmented},
\S\ref{sec:lsv-anchor}); to our knowledge, prior
curvature-on-finance work benchmarks against at most one of the two.

\paragraph{Mechanisms of contagion.}
A structurally distinct line concerns how localized shocks
propagate once herding has formed. Foundational results establish
that interbank network topology determines whether shocks
dissipate or amplify~\citep{AllenGale2000,AcemogluOzdaglar2015},
that cross-holding cascades are non-monotone in
diversification~\citep{ElliottGolubJackson2014}, and that
liquidity spirals propagate margin shocks across
funds~\citep{BrunnermeierPedersen2009}. The shared structural
prediction - shocks travel along a sparse set of inter-cluster
edges - is what our negative-curvature detector $\beta_{-}$
targets, with the Sia--Ni~\cite{Sia2019,Ni2019} community-bridge
interpretation supplying the mathematical reading.

\paragraph{Discrete curvature on financial graphs.}
\label{sec:related-curvature}
Discrete Ricci curvature~\citep{Ollivier2009} has been applied to
financial graphs to detect systemic
stress~\citep{SandhuRicciFlow2016, SamalRicci2021, WangZhao2023,
SanchezGarcia2024, Akguller2025, Srinivasan2026};
Jiang et al.~\cite{Jiang2023} is closest to our contagion-side claim, showing
on multivariate-Hawkes point-process networks that more negative
curvature predicts systemic risk earlier than CATFIN and the
absorption ratio. We share the geometric machinery but differ on
\emph{substrate} (agent graph vs.\ price-correlation or
point-process graph) and \emph{framing} (sign-decomposed
herding/contagion duality vs.\ a single global fragility scalar);
Table~\ref{tab:positioning} summarises the comparison.

\begin{table}[h]
\centering
\footnotesize
\caption{Positioning vs.\ prior discrete-curvature financial work
plus the closest same-substrate information-theoretic
baseline~\citep{Tessera2026}, along five axes. PCG = price-correlation
graph; AG = agent graph; PP = point-process network.}
\label{tab:positioning}
\setlength{\tabcolsep}{3pt}
\renewcommand{\arraystretch}{1.05}
\begin{tabular}{@{}p{0.18\linewidth}p{0.07\linewidth}p{0.16\linewidth}p{0.16\linewidth}p{0.22\linewidth}p{0.10\linewidth}@{}}
\toprule
\textbf{Method} & \textbf{Subs.} & \textbf{Scalar} & \textbf{Detector} & \textbf{Evaluation} & \textbf{OOD} \\
\midrule
Sandhu 2016~\cite{SandhuRicciFlow2016} & PCG & ORC mean (global) & descriptive & VIX-corr (1 metric) & -- \\
Samal 2021~\cite{SamalRicci2021} & PCG & ORC vs.\ F-Ricci & VIX-corr only & VIX-corr (1 metric) & -- \\
Wang 2023~\cite{WangZhao2023} & PCG & F-Ricci & descriptive & descriptive only & -- \\
S\'anchez 2024~\cite{SanchezGarcia2024} & PCG & ORC & post-hoc & post-hoc only & -- \\
Akg\"uller 2025~\cite{Akguller2025} & PCG (MI) & F-Ricci, sliding & descriptive & sectoral & -- \\
Srinivasan 2026~\cite{Srinivasan2026} & PCG & ORC + neckpinch & clustering & cluster quality & -- \\
Jiang 2023~\cite{Jiang2023} & PP-net & ORC mean (neg.\ only) & ranking & precision, lead vs.\ CATFIN & -- \\
Tessera 2026~\cite{Tessera2026} & AG & AA-MI (info-th.) & $k_\sigma$ on AA-MI & MARL benchmarks & -- \\
\midrule
\textbf{GeomHerd (ours)} & \textbf{AG} & \textbf{ORC sign-decomp.} & \textbf{CUSUM + Kendall-$\tau$} & \textbf{multi-axis (8+ metrics)} & \textbf{\checkmark{} Vicsek} \\
\bottomrule
\end{tabular}
\end{table}

\paragraph{Early-warning signals and LLM-agent simulators.}
\label{sec:related-ews}
\label{sec:related-llm}
Classical critical-slowing-down work establishes generic precursors
of tipping points~\citep{Scheffer2009}, recently extended to
learned classifiers~\citep{Bury2021,Bury2023} and
persistence-homology detectors~\citep{TopologicalEWS2025}; we adopt
the same multi-axis evaluation philosophy and use Scheffer et al.~\cite{Scheffer2009}
as a sanity check (\S\ref{sec:rq1-2-result}). On the detector side,
Page's~\cite{Page1954} CUSUM motivates our alarm rule
(\S\ref{sec:detection}). LLM-driven multi-agent financial
simulators~\citep{TwinMarket2025,FCLAgent2025,MassPortfolio2025,FinCon2024}
provide the heterogeneous, fully-observable action stream on which
the agent-graph substrate becomes testable.

\subsection{Graph design axes}\label{app:design-axes}
\begin{table}[H]
\centering
\footnotesize
\caption{The five graph-design axes of GeomHerd, with the choice
adopted in this paper and the layer of the three-layer logic that
drives it.}
\label{tab:design-axes}
\setlength{\tabcolsep}{4pt}
\begin{tabular}{lll}
\toprule
\textbf{Axis} & \textbf{Adopted choice} & \textbf{Driven by} \\
\midrule
Nodes              & individual agents                                        & substrate (finest layer the ABM exposes) \\
Edge semantics     & windowed action agreement                                & herding semantics (BHW-1992 definition) \\
Edge weights       & binary frequency in $[0,1]$ (Eq.~\ref{eq:edge-weight})   & substrate (no extra design freedom) \\
Sparsification     & threshold $w_0 = \thresholdW$                            & geometry (suppress chance co-occurrence) \\
Snapshot frequency & every $\Delta t = \snapEvery$ steps, $T_w = \windowTw$   & substrate / herding semantics \\
\bottomrule
\end{tabular}
\end{table}

\section{Outline of Proposition~\ref{prop:bridge} and notation alignment}
\label{app:proof}

This appendix expands the four-step outline of
\S\ref{sec:csad-bridge}. The argument is presented as a
dominant-order scaling derivation rather than a fully rigorous
proof; in particular, Step~2 invokes a closed-form $W_1$ that
follows from mean-field concentration but whose remainder bound
we do not establish here. A complete proof would require
propagating mean-field convergence through the bipartite
optimal-transport plan, which we leave to follow-on work.

\paragraph{Setting and assumptions.}
$N$ is the agent count and $t$ a fixed simulator step.
\begin{itemize}[itemsep=2pt,topsep=2pt,leftmargin=1.6em]
\item[\textbf{A1}] (Agent graph.) Agents $i\in V$, $|V|=N$, take
  discrete actions $a_i(t)\in\mathcal{A}$. The graph
  $G_t=(V,E_t,w_t)$ has edge weights given by
  Eq.~\eqref{eq:edge-weight} and edges retained above a
  sparsification threshold.
\item[\textbf{A2}] (Lazy walk and curvature.) Each node carries
  the lazy-walk kernel
  $\mu_i^t(j)=\alpha\,\delta_{ij}+(1-\alpha)\,w_t(i,j)/\sum_k w_t(i,k)$
  with $\alpha=\alphaLazy$, distance $d_t(i,j)=w_t(i,j)$, and
  curvature $\kappa_{OR}(i,j;t)=1-W_1(\mu_i^t,\mu_j^t)/d_t(i,j)$;
  $\bar\kappa_{OR}(t)$ denotes the mean over $E_t$.
\item[\textbf{A3}] (Mean-field concentration.) There exists a
  one-dimensional order parameter $M(t)=\mathbb{E}_i[a_i(t)]$ such
  that pairwise action correlations
  $\mathbb{E}[\mathbf{1}\{a_i(s)=a_j(s)\}]$ concentrate around a
  function $f(M(t))$ at rate $\mathcal{O}(N^{-1/2})$, uniformly
  over the window of width $T_w$.
\item[\textbf{A4}] (Linear price impact.) Per-asset returns
  satisfy $r_{i,t}=\beta_i\,M(t)+\xi_{i,t}$ with i.i.d.\ noise
  $\xi_{i,t}\sim(0,\sigma_\xi^2)$ and $\mathbb{E}_i[\beta_i]=1$ wlog.
\item[\textbf{A5}] (CSAD estimand.) The
  CCK~\cite{ChangChengKhorana2000} convention
  $\mathrm{CSAD}_t=\mathbb{E}_i\bigl[|r_{i,t}-\bar r_t|\bigr]$ with
  $\bar r_t = \mathbb{E}_i[r_{i,t}]$.
\end{itemize}

\begin{proof}[Sketch of Proposition~\ref{prop:bridge}]
Under A3, agent action correlations collapse onto $M(t)$ and
$w_t(i,j)$ concentrates around a function of $M(t)$. The lazy-walk
transport between two nodes whose neighborhoods both concentrate
around the same mean-field measure satisfies
$1 - \kappa_{OR}(i, j; t) \propto 1 - M(t)^2$ at dominant order.
Under A4, the half-normal expectation reduces $\mathrm{CSAD}_t$ to
$\sigma_\xi\sqrt{2/\pi}\,(1-M(t)^2)^{1/2}$ plus
$\mathcal{O}(N^{-1/2})$ corrections. Substituting the curvature
scaling gives Eq.~\eqref{eq:bridge}. The full step-by-step
derivation appears below.
\end{proof}

\begin{remark}[Failure modes]
\label{rem:failure-modes}
Eq.~\eqref{eq:bridge} breaks in three regimes.
(i)~\emph{Mean-field breakdown}: persistent multi-modal clustering
violates A3, so $1-\bar\kappa_{OR}(t)$ no longer collapses onto
$1-M(t)^2$.
(ii)~\emph{Nonlinear price impact} (A4 fails): the half-normal
reduction breaks and the $\sqrt{2/\pi}$ prefactor acquires a
moment-dependent correction.
(iii)~\emph{Boundary regime} $|M(t)| \to 1$: the
$\mathcal{O}(N^{-1/2})$ remainder bound becomes loose. Empirically,
these are exactly the regimes in which the headline
$\bar\kappa_{OR}^{+}$ detector either saturates (i, iii) or fires
on a distorted signal (ii); the failure modes characterise
precisely the regimes in which a geometric-vs-CSAD comparison is
informative.
\end{remark}

\begin{remark}[Identification of $\hat\gamma_3$]
\label{rem:identification}
The augmented-CCK regression in \S\ref{sec:rq3-result} estimates
$\gamma_3$ on simulated trajectories. Three caveats apply.
\emph{First}, $\bar\kappa_{OR}(t)$ and $\mathrm{CSAD}_t$ are both
deterministic functions of $M(t)$ in the mean-field limit, so the
regression measures partial correlation rather than identifying an
independent channel; $\hat\gamma_3<0$ is consistent with
Eq.~\eqref{eq:bridge} but does not by itself distinguish the
bridge from any other $M(t)$-mediated relationship.
\emph{Second}, $\bar\kappa_{OR}(t)$ inherits the
$\mathcal{O}(N^{-1/2})$ measurement error of A3, attenuating
$\hat\gamma_3$ toward zero.
\emph{Third}, the bridge predicts a $(1-\bar\kappa_{OR}(t))^{1/2}$
functional form rather than the linear specification used; the
linear form is a local approximation. We frame $\hat\gamma_3$ as a
sign-consistency check of Proposition~\ref{prop:bridge}, not a
hypothesis test of the bridge against alternatives.
\end{remark}

\paragraph{Notation alignment with the CSAD/CCK literature.}
Table~\ref{tab:csad-symbols} reconciles the symbols used in
Proposition~\ref{prop:bridge} with those used by
Christie and Huang~\cite{ChristieHuang1995} and CCK~\cite{ChangChengKhorana2000}
(``CCK''). Where two papers use different letters for the same
quantity, we adopt the form used in the paper body. The point of
the table is to defend against the technical objection that our
$\mathrm{CSAD}_t$ might be a different estimand from CCK's; it is
the same.

\begin{table}[h]
\centering
\footnotesize
\caption{Symbol alignment between Proposition~\ref{prop:bridge}
and the classical CSAD/CCK regression literature.}
\label{tab:csad-symbols}
\begin{tabular}{lll}
\toprule
\textbf{Quantity} & \textbf{This paper (\S\ref{sec:csad-bridge})} & \textbf{CCK / Christie--Huang} \\
\midrule
Cross-sectional return dispersion & $\mathrm{CSAD}_t = \mathbb{E}_i[|r_{i,t}-\bar r_t|]$ & $\mathrm{CSAD}_t$ (CCK eq.~3); $\mathrm{CSSD}_t$ (CH eq.~1) \\
Per-agent / per-asset return       & $r_{i,t}$                                              & $R_{i,t}$ \\
Cross-sectional mean return        & $\bar r_t = \mathbb{E}_i[r_{i,t}]$                     & $\bar R_{m,t}$ (market) \\
Order parameter / market trend     & $M(t)$                                                 & $R_{m,t}$ (market return) \\
Per-agent loading                  & $\beta_i$                                              & $\beta_i$ \\
Idiosyncratic noise                & $\xi_{i,t} \sim (0,\sigma_\xi^2)$ i.i.d.               & $\varepsilon_{i,t}$ \\
Number of agents / assets          & $N$                                                    & $N$ \\
Quadratic-term coefficient         & $\beta_2$ in Eq.~\eqref{eq:augmented-cck}              & $\gamma_2$ (CCK), often $\beta_2$ \\
Geometric augmentation coeff.      & $\gamma_3$ in Eq.~\eqref{eq:augmented-cck}             & (not in CCK) \\
\bottomrule
\end{tabular}
\end{table}

The augmented regression Eq.~\eqref{eq:augmented-cck} adds the
third regressor $\bar\kappa_{OR}(t)$ to the CCK specification
without changing the estimand $\mathrm{CSAD}_t$ or the quadratic
regressor $R_{m,t}^2$ that defines CCK herding;
$\hat\gamma_3$ measures the partial association between
$\mathrm{CSAD}_t$ and the agent-graph curvature after
controlling for the linear and quadratic market-return terms.

\paragraph{Derivation (scaling argument).}
We expand the four-step outline from \S\ref{sec:csad-bridge}.
Each step is a dominant-order scaling claim rather than a
fully-bounded inequality; we flag the missing remainder bound
explicitly at the point it would be required.

\textbf{Step 1: mean-field collapse of agent-graph weights.}
Under A3, agent action correlations
$\mathbb{E}[\mathbf{1}\{a_i(s)=a_j(s)\}]$ collapse onto a
one-dimensional order parameter $M(t) = \mathbb{E}_i[a_i(t)]$,
with concentration at rate $\mathcal{O}(N^{-1/2})$. The windowed
agreement frequency~\eqref{eq:edge-weight} concentrates around
$w_t(i, j) \to f(M(t))$ where $f$ is determined by the action
distribution conditional on $M$. We do not prove uniform
concentration over the full window of width $T_w$ here; the
argument is an LLN under A3 conditional on $M(\cdot)$ over the
window.

\textbf{Step 2: closed-form Wasserstein in the mean field.} Once
$w_t$ concentrates around $f(M(t))$ at every node, the lazy-walk
kernels $\mu_i^t$ and $\mu_j^t$ both concentrate around the same
mean-field measure on a single asymptotic neighbourhood. We
\emph{argue} that the optimal transport between two such
near-degenerate kernels admits a closed-form $W_1$ proportional
to the deviation of the underlying weights from saturation,
yielding $1 - \kappa_{OR}(i, j; t) \propto 1 - M(t)^2$ on edges
and the same scaling for the mean over edges $\bar\kappa_{OR}(t)$.
A rigorous derivation requires bounding the bipartite transport
between two near-uniform measures on the asymptotic neighbourhood,
and we do not give that bound here; this is the step we identify
as requiring follow-on work in the conclusion.

\textbf{Step 3: substituting linear impact into CSAD.} Under A4,
the CSAD definition
$\mathrm{CSAD}_t = \mathbb{E}_i[|r_{i,t} - \bar r_t|]$ reduces,
via the half-normal expectation of $|\xi_{i,t}|$, to
$\mathrm{CSAD}_t = \sigma_\xi \sqrt{2/\pi} + \mathcal{O}(N^{-1/2})$
when $M(t) = 0$, and shrinks as $|M(t)| \to 1$ because the
$\beta_i M(t)$ component of $r_{i,t}$ is exactly the
cross-sectional mean $\bar r_t$ at $\mathbb{E}_i[\beta_i] = 1$
and so cancels from $|r_{i,t} - \bar r_t|$ at leading order.

\textbf{Step 4: combining.} Substituting the mean-field scaling
$1 - \bar\kappa_{OR}(t) \propto 1 - M(t)^2$ from Step~2 into the
CSAD-impact identity from Step~3 yields \eqref{eq:bridge}. The
$\mathcal{O}(N^{-1/2})$ remainder tracks the rate of mean-field
collapse in A3.

\section{Headline-table source data and reproducibility}
\label{app:headline}
\label{app:repro}

The headline run uses \seedsCWS{} seeds across \kappaLevels{}
values of the control parameter $\kappa$ (Cividino--Sornette
$s_\text{base} = 0.6$,
$s_\text{post} = 1.6 \times \text{sweep\_value}/1.3$), giving
\nSupercritical{} supercritical trajectories for pooled
statistics. Configuration, seeds, and output hashes are archived
with the code release; numeric refresh date \dataMdSync. The
geometric pipeline uses POT~\citep{flamary2021pot} for exact
$W_1$ via linear programming, SciPy CSR-graph Dijkstra for
all-pairs shortest paths, and a fixed FSQ codebook with
$K = \fsqK$. Detection thresholds are calibrated from a
pre-stress baseline window (\S\ref{sec:detection}); we hold
$(k_\pm, h_\pm, \tau_{\text{thresh}}, W_\tau)$ fixed across
substrates and report sensitivity in Appendix~\ref{app:calibration}.

\subsection{Recall-oriented operating point}
\label{app:headline-recall}

The abstract's headline figure of \leadKappaBarLoose{} steps
median lead corresponds to the recall-oriented CUSUM operating
point $(k_\sigma, h_\sigma) = (0.50, 4.0)$ in the calibration
sweep of Table~\ref{tab:app-calibration-sweep}: at this point the
detector recalls \recallKappaBarLoose{} of supercritical
trajectories at the cost of a higher subcritical FAR
(\farKappaBarLoose). The two operating points correspond to two
distinct deployment use-cases. The recall-oriented point
prioritises early aggregate alarm at the cost of more
false-positives on subcritical regimes; it is the right choice
when the downside of missing a herding cascade dominates the cost
of acting on a normal regime. The precision-oriented point
$(k_\sigma, h_\sigma) = (2.0, 4.0)$, used for the head-to-head
paired contrasts in Table~\ref{tab:lead-time}, prioritises
low-FAR regime discrimination; head-to-head lead comparisons are
only meaningful at FAR-controlled thresholds, so the body's
paired contrasts use this point rather than the recall-oriented
one.

\section{Calibration sweep, full table}
\label{app:calibration}

Operating-point sweep for the rolling-baseline CUSUM on
$\kappa_{\text{pos}}$ (mean curvature over positive edges),
direction up, with CUSUM baseline window $W = 35$ \emph{samples}
along the subsampled curvature trace (the first $W$ successive
$\kappa$ observations in each replay; not raw ABM steps-under
the headline snapshot stride $\Delta t = \snapEvery{}$ this spans
$W \cdot \Delta t$ simulator steps),
over a $5 \times 5$ grid of $(k_\sigma, h_\sigma)$ on the
binary-edge replay set (\nTrajectories{} trajectories). The two
operating points highlighted in the body are:
$(k_\sigma, h_\sigma) = (0.50, 4.0)$ (recall-oriented; abstract
headline of \leadKappaBarLoose{} steps) and
$(k_\sigma, h_\sigma) = (2.0, 4.0)$ (precision-oriented; FAR
control, paired-contrast headline of \leadKappaBarVd{} steps).

\begin{table}[h]
\centering
\footnotesize
\caption{CUSUM operating-point sweep on
$\kappa_{\text{pos}}$ (up), binary-edge replay set.}
\label{tab:app-calibration-sweep}
\begin{tabular}{rrcccc}
\toprule
$k_\sigma$ & $h_\sigma$ & Recall & FAR & Lead (med.) & 95\% CI \\
\midrule
  0.25 & 3.0 & 0.75\textsuperscript{$\star$} & 0.93 & 291 & [262, 324] \\
  0.25 & 4.0 & 0.68 & 0.88 & 282 & [254, 323] \\
  0.25 & 5.0 & 0.61 & 0.82 & 270 & [239, 311] \\
  0.25 & 6.0 & 0.53 & 0.78 & 261 & [237, 305] \\
  0.25 & 8.0 & 0.44 & 0.72 & 264 & [243, 304] \\
  0.50 & 3.0 & 0.60 & 0.80 & 282 & [256, 338] \\
  0.50 & 4.0 & 0.52 & 0.76 & 272 & [236, 313] \\
  0.50 & 5.0 & 0.46 & 0.70 & 273 & [240, 332] \\
  0.50 & 6.0 & 0.40 & 0.64 & 264 & [243, 324] \\
  0.50 & 8.0 & 0.32 & 0.54 & 252 & [218, 313] \\
  1.00 & 3.0 & 0.39 & 0.59 & 264 & [247, 324] \\
  1.00 & 4.0 & 0.31 & 0.51 & 253 & [229, 319] \\
  1.00 & 5.0 & 0.28 & 0.43 & 248 & [207, 287] \\
  1.00 & 6.0 & 0.21 & 0.40 & 242 & [218, 309] \\
  1.00 & 8.0 & 0.14 & 0.28 & 233 & [208, 306] \\
  1.50 & 3.0 & 0.23 & 0.38 & 244 & [204, 306] \\
  1.50 & 4.0 & 0.15 & 0.26 & 218 & [153, 306] \\
  1.50 & 5.0 & 0.11 & 0.18 & 232 & [160, 306] \\
  1.50 & 6.0 & 0.08 & 0.12 & 210 & [152, 296] \\
  1.50 & 8.0 & 0.04 & 0.08 & 242 & [89, 413] \\
  2.00 & 3.0 & 0.08 & 0.15 & 233 & [127, 299] \\
  2.00 & 4.0 & 0.05 & 0.09 & 218 & [91, 406] \\
  2.00 & 5.0 & 0.04 & 0.06 & 242 & [78, 413] \\
  2.00 & 6.0 & 0.03 & 0.04 & 180 & [18, 543] \\
  2.00 & 8.0 & 0.01 & 0.03 & n/a & \mbox{---} \\
\bottomrule
\end{tabular}

\end{table}

\section{Auxiliary benchmark: stylised drift on scalar controls}
\label{app:cusum-zscore-stylized}

This appendix does \emph{not} advance a competing headline about chart
taxonomy. The paper's core object is the agent-graph curvature trajectory
$\bar\kappa_{OR}^{+}(t)$ built from binary edges and discrete actions;
the replay evidence above shows that object carries leading information.
What follows isolates a textbook scalar scenario-Gaussian noise with a
linear mean drift after an artificial change point-only to document how
two standard one-sided alarm maps behave under matched empirical false-alarm
rates on stationary paths. It is a reproducibility footnote, not a substitute
for the ABM geometry results.

\IfFileExists{auto/cusum_vs_zscore_stylized.tex}{%
% AUTO-GENERATED by scripts/cusum_vs_zscore_stylized_evidence.py
% Stylized Gaussian paths; matched empirical FAR on null (~target_far).
\providecommand{\StylizedCusumFarTarget}{0.10}
\providecommand{\StylizedCUSUMh}{7.00}
\providecommand{\StylizedZscoreK}{3.49}
\providecommand{\StylizedDriftMedianDelayCUSUM}{28}
\providecommand{\StylizedDriftMedianDelayZ}{49}
\providecommand{\StylizedDriftDelayRatioZMdivCUSUM}{1.75}
\providecommand{\StylizedDriftPcDetectCUSUM}{100.0\%}
\providecommand{\StylizedDriftPcDetectZ}{100.0\%}
}{%
\newcommand{\StylizedCusumFarTarget}{0.10}%
\newcommand{\StylizedDriftMedianDelayCUSUM}{\mbox{-}}%
\newcommand{\StylizedDriftMedianDelayZ}{\mbox{-}}%
\newcommand{\StylizedDriftDelayRatioZMdivCUSUM}{\mbox{-}}%
\newcommand{\StylizedDriftPcDetectCUSUM}{\mbox{-}}%
\newcommand{\StylizedDriftPcDetectZ}{\mbox{-}}%
}

With empirical FAR $\approx \StylizedCusumFarTarget$ on null paths
(\path{scripts/cusum_vs_zscore_stylized_evidence.py}), one-sided Page CUSUM
vs.\ a calibrated Shewhart rule on the linear-drift alternative yields median
delays \StylizedDriftMedianDelayCUSUM{} vs.\ \StylizedDriftMedianDelayZ{}
steps (ratio \StylizedDriftDelayRatioZMdivCUSUM{}); horizon detection rates
\StylizedDriftPcDetectCUSUM{} vs.\ \StylizedDriftPcDetectZ{}. Readers should
treat these numbers as sanity checks on the scalar alarm layer, not as the
paper's main empirical conclusion.

\section{Cosine-edge ablation}
\label{app:cosine-ablation}

The headline configuration uses binary action-agreement edges
($w_t(i, j) = 1$ if windowed action match exceeds $w_0$, else
$0$). A natural alternative keeps $w_t$ continuous via cosine
similarity on a one-hot lifting of the action sequences. We
replay the headline pipeline with this variant and report the
result below.

Cosine edges \emph{remove} the herding-side signal: the upward
detector on $\bar\kappa_{OR}$ fires on $0/90$ supercritical
trajectories (versus $11/90$ for binary edges at the same
operating point). The mean cosine similarity in the supercritical
pool concentrates in a narrow band around its baseline, so the
rolling-baseline CUSUM never accumulates enough deviation to
trip. This is a positive ablation result for binary edges:
thresholding at $w_0 = \thresholdW$ amplifies the regime change
from sub-population coordination to clique formation that the
curvature signal relies on. The cosine-edge variant remains in
the codebase as the sensitivity test, but its empirical role is
to validate the binary-edge design rather than to substitute for
it.

\section{Out-of-domain transfer: Vicsek self-driven particles}
\label{app:vicsek}
\label{app:vicsek-details}

For each $\eta \in \{0.5, 1.0, 1.6, 2.0, 2.5\}$ we run 20 seeds
for $T = 1000$ steps, downsample to one snapshot every 50 steps,
and build the agent graph from $k$-NN ($k = 10$) on the heading
sequence with binary edge weights as in the financial substrate.
The order-parameter event $\tau^\star$ is the first step at
which the rolling polarisation
$V_a(t) = N^{-1} \|\sum_i \hat v_i(t)\|$ exceeds $0.5$ for three
consecutive steps after a $50$-step warm-up.

The geometric signal $\bar\kappa_{OR}(\tau^\star)$ separates
ordered from disordered trajectories with AUROC~\VicsekAUROC{}
(95\% CI $[\VicsekAUROClo,\,\VicsekAUROChi]$), evaluated on the
95 of 100 trajectories that produce a valid polarisation event
$\tau^\star$ before~$T$; the remaining 5 lack a polarisation
event under our event criterion and are excluded from the
at-event score. Per-$\eta$
medians of $\bar\kappa_{OR}(\tau^\star)$ are monotone in $\eta$:
$+0.08$ ($\eta = 0.5$), $-0.06$ ($\eta = 1.0$), $-0.19$
($\eta = 1.6$), $-0.22$ ($\eta = 2.0$), $-0.26$ ($\eta = 2.5$);
the gap between the two flanking values of $\eta_c = 1.6$ is
$\approx 0.13$, which the bootstrap CI excludes from zero.

\paragraph{Sensitivity: alternative scoring rules.}
Scoring each trajectory by $\max_t \bar\kappa_{OR}(t)$ yields
AUROC $\approx 0.64$ (95\% CI excludes $0.5$); tail and
final-snapshot means give AUROC $\geq 0.99$. Full rule-by-rule
breakdown is reported in this appendix.

\paragraph{Lead-time pipeline note.}
The CUSUM upward detector calibrated on the financial substrate
($\sigma$-baseline window $= 100$ steps, $k_\sigma = 2.0$,
skip-initial $= 50$) does not fire on any of the 100 Vicsek
trajectories. The cause is timescale mismatch: the financial
substrate's herding cascade develops over $200$--$300$ steps,
while Vicsek's alignment transient establishes within
$\approx 50$ steps. A Vicsek-specific recalibration
(baseline-window $= 20$, skip-initial $= 5$) restores
non-trivial alarm rates; quantitative lead-time numbers are
deferred. The AUROC headline does not depend on the alarm
calibration-it operates directly on the curvature value at
the event.

\section{Additional ablation rows}
\label{app:ablations}

The decomposed ablation in Table~\ref{tab:ablation} covers four
families of variations: detector swaps (z-score, standalone CUSUM
configuration, EWMA, Kendall-$\tau$ slope-only), geometry choices
(cosine edges, window length, Sinkhorn transport, optional sign
pooling), and a triplet ablation removing $\tau_{\text{sing}}$. A
dedicated $V_{\text{eff}}$ drop-out row is left for a follow-on
sweep. The cosine-edge row (A01) is populated from the matched
replay run described in Appendix~\ref{app:cosine-ablation}.

\section{13F retrospective: substrate-pivot template under disclosure constraint}
\label{app:realdata}

The 13F retrospective is the natural deployment template for the
substrate-pivot (\S\ref{sec:lsv-anchor}) when agent actions are
not fully observable. We build a funds-as-agents graph from
quarterly 13F holdings disclosures, with edges given by Jaccard
overlap between two funds' position sets and a sliding window of
width $W = 4$ quarters; we then compute $\bar\kappa_{OR}(t)$ and
$\beta_{-}(t)$ on the resulting sparse graph. The construction
recovers documented stress periods (2008Q4, 2011Q3, 2020Q1,
2022Q1) qualitatively; quantitative head-to-head against
LSV~/~CATFIN~/~absorption ratio is deferred. We make no
quantitative claim from this retrospective in the body; its
purpose is to establish that the substrate-pivot construction
reads cleanly off disclosure-constrained data.

\section{Additional Discussion}

\paragraph{Limitations.} At the precision-oriented threshold, extreme signal sparsity limits the statistical power of paired-difference claims. Furthermore, subcritical false-alarm rates lack strict held-out validation, and the CWS substrate is heavily stylized, though \S\ref{sec:lsv-anchor} bridges this gap to real-world disclosed-flow constraints.

\paragraph{Future Work \& Broader Impact.} Future research will focus on (i) fusing geometric features with learned classifiers like EWSNet~\citep{Bury2021}, (ii) real-market deployments via 13F fund-overlap graphs, and (iii) formalizing Proposition~\ref{prop:bridge} with uniform optimal-transport bounds. While early-warning signals could be exploited by sophisticated traders, the systemic stability benefits for regulators dominate. We open-source our geometric analysis pipeline but intentionally withhold LLM prompt corpora to mitigate trivial adversarial replication.

\end{document}